\documentclass[11pt,a4paper]{article} 
\usepackage{fullpage}
\usepackage{amsmath, amssymb} 
\usepackage{epsf,graphicx}
\usepackage{hyperref} 
\usepackage{color}
\usepackage[nodayofweek]{datetime} 
\usepackage{enumerate}
\usepackage[displaymath,mathlines]{lineno} 
\usepackage[T1]{fontenc}
\usepackage{etoolbox} 
\usepackage{pdfpages}

\newcommand{\subchain}[2]{\ensuremath{C[#1, #2]}}
\newcommand{\subpolygon}[2]{\ensuremath{P[#1, #2]}}
\newcommand{\radrestricted}{\ensuremath{r_{ij}^*}}

\newcommand{\maxrad}[2]{\ensuremath{r_\mathrm{max}{(#1,#2)}}}
\newcommand{\radbd}[2]{\ensuremath{r({#1}, {#2})}}
\newcommand{\rad}[1]{\ensuremath{r(#1)}}
\newcommand{\spm}[1]{\mathsf{SPM}_{#1}} 
\newcommand{\pu}{\alpha}
\newcommand{\pw}{\beta} 

\newcommand{\sarcs}{\ensuremath{S}}
\newcommand{\interior}[1]{\ensuremath{\mathrm{int}(#1)}}
\newcommand{\bd}[1]{\ensuremath{\partial #1}}
\newcommand{\vcw}[1]{\ensuremath{v_{cw}(#1)}}
\newcommand{\vccw}[1]{\ensuremath{v_{ccw}(#1)}}
\newtheorem{theorem}{Theorem} 
\newtheorem{lemma}[theorem]{Lemma}

\newtheorem{corollary}[theorem]{Corollary}

\newcommand{\eps}{\varepsilon}

\newbox\ProofSym \setbox\ProofSym=\hbox{%
  \unitlength=0.18ex%
  \begin{picture}(10,10) \put(0,0){\framebox(9,9){}}
    \put(0,3){\framebox(6,6){}}
  \end{picture}}

\title{The geodesic $2$-center problem in a simple polygon\thanks{This
    work was supported by the NRF grant 2011-0030044 (SRC-GAIA) funded
    by the government of Korea.}} 

\author{Eunjin Oh\thanks{Pohang University of Science and Technology,
    Korea. Email: {\tt{\{jin9082, heekap\}@postech.ac.kr}}} \and
  Jean-Lou De Carufel\thanks{University of Ottawa, Canada.  Email:
    {\tt{jdecaruf@uottawa.ca}}} \and Hee-Kap Ahn\footnotemark[2]~\thanks{Corresponding author.} }

\graphicspath{{figures/}}
\date{}
\begin{document}
\maketitle
\begin{abstract}
  The \emph{geodesic $k$-center problem in a simple polygon} with $n$
  vertices consists in the following.  Find a set $S$ of $k$ points in
  the polygon that minimizes the maximum geodesic distance from any
  point of the polygon to its closest point in $S$.  In this paper, we
  focus on the case where $k=2$ and present an exact algorithm that
  returns a geodesic $2$-center in $O(n^2\log^2 n)$ time.
\end{abstract}

\section{Introduction}
The \emph{geodesic $k$-center problem in a simple polygon $P$} with
$n$ vertices consists in the following.  Find a set $S$ of $k$ points
in $P$ that minimizes
\begin{linenomath}
$$\max_{p \in P} \min_{s \in S} d(s,p) ,$$
\end{linenomath}
where $d(x,y)$ is the length of the shortest path between $x$ and $y$
lying in $P$ (also called \emph{geodesic distance}).  The set $S$ is
called a \emph{$k$-center of $P$}.  Geometrically, this is equivalent
to find $k$ smallest-radius geodesic disks with the same
radius whose union contains $P$.
 
The \emph{$2$-dimensional Euclidean $k$-center problem} is similar to
the geodesic $k$-center problem in a simple polygon $P$.  The only
difference is that in the Euclidean $k$-center problem, the distance
between two points $x$ and $y$ is their Euclidean distance, denoted by
$\|x-y\|$.  That is, given a set $\mathcal{P}$ of $n$ points in the
plane, find a set $S$ of $k$ points in $\mathbb{R}^2$ that minimizes
\begin{linenomath}
$$\max_{p \in \mathcal{P}} \min _{s \in S} \|p-s\| .$$
\end{linenomath}

Computing a $k$-center of points is a typical problem in clustering.
Clustering is the task of partitioning a given set into subsets subject
to various objective functions, which have applications
in pattern-analysis, decision-making and
machine-learning situations including data mining, document retrieval, and
pattern classification~\cite{clustering-reviw}.
The Euclidean $k$-center problem has been studied extensively.  The
$1$-center of $\mathcal{P}$ coincides with the center of the minimum enclosing
circle of $\mathcal{P}$, which can be computed in linear
time~\cite{DBLP:journals/siamcomp/Megiddo83a}.  Chan showed that the
$2$-center problem can be solved in $O(n \log^2 n \log^2 \log n)$
deterministic time~\cite{Euclidean-2-center}.  The $k$-center problem
can be solved in $O(n^{O(\sqrt{k})})$ time~\cite{Euclidean-k-center}.
It is NP-hard to approximate the Euclidean $k$-center problem
within an approximation factor smaller than
1.822~\cite{FederGreene88}.  Kim and Shin presented an $O(n \log^3 n
\log \log n)$-time algorithm for computing two congruent disks whose
union contains a convex $n$-gon~\cite{convex-center}.

The $1$-center problem has also been studied under the geodesic metric
inside a simple polygon.  Asano and Toussaint presented the first
algorithm for computing the geodesic $1$-center of a simple polygon
with $n$ vertices in $O(n^4 \log n)$ time~\cite{1-center-first}.  In
1989, the running time was improved to $O(n \log n)$ time by Pollack
et al.~\cite{1-center1989}.  Their technique can be described as
follows.  They first triangulate the polygon and find the triangle $T$
that contains the center in $O(n\log n)$ time.   Then they
  subdivide $T$ further and find a region containing the center
  such that the combinatorial structures of the geodesic paths from
  each vertex of $P$ to all points in that region are the same.  Finally,
the problem is reduced to find the lowest point of the upper envelope
of a family of distance functions in the region, which can be done in linear time
using a technique by Megiddo~\cite{megiddo-minmax}.  Recently, the
running time for computing the geodesic $1$-center was improved to linear
by Ahn et al.~\cite{1-center-SoCG,1-center}, which is optimal.  In their paper,
instead of triangulating the polygon, they construct a set of $O(n)$
chords.
Then they recursively subdivide the polygon into $O(1)$
cells by a constant number of chords and find the cell containing
the center.  Finally, they obtain a triangle
containing the center.  In this triangle, they find the lowest point of
the upper envelope of a family of functions, which is the geodesic 1-center
of the polygon, using an
algorithm similar to the one of Megiddo~\cite{megiddo-minmax}.

Surprisingly, there has been no result for the geodesic $k$-center
problem for $k > 1$, except the one by Vigan~\cite{Vigan2013}.  They
gave an exact algorithm for computing a geodesic $2$-center
in a simple polygon with $n$ vertices, which runs in $O(n^8 \log n)$
time.  The algorithm follows the framework of Kim and
Shin~\cite{convex-center}.
However, the algorithm does not seem to work as it is because of the
following reasons.  They claim that the decision version of the
geodesic $2$-center problem in a simple polygon can be solved using a
technique similar to the one by Kim and Shin~\cite{convex-center}
without providing any detailed argument.  
They apply parametric search using their decision algorithm, but
they do not describe how their parallel algorithm works.  The
parallel algorithm by Kim and Shin does not seem
to extend for this problem.

\subsection{Our results}
\label{section our results}
In this paper, we present an $O(n^2\log^2n)$-time algorithm that
solves the geodesic $2$-center problem in a simple polygon with $n$
vertices.  The main steps of our algorithm can be described as
follows.  We first observe that a simple polygon $P$ can always be
partitioned into two regions by a geodesic path $\pi(x,y)$ such that
\begin{itemize}
\item $x$ and $y$ are two points on the boundary of $P$, and
\item the set consisting of the geodesic 1-centers of the two regions
of $P$ defined by $\pi(x,y)$ is a geodesic 2-center of $P$.
\end{itemize}
Then we consider $O(n)$ candidate pairs of edges of $P$, one of which,
namely $(e,e')$, satisfies $x\in e$ and $y\in e'$.  We explain how to
find these candidate pairs of edges in $O(n^2\log n)$ time.  Finally, we
present an algorithm that computes a $2$-center restricted to
such a pair of edges in $O(n\log^2n)$ time using parametric
search~\cite{parametric} with a decision algorithm and a parallel
algorithm.

\section{Preliminary}
A polygon $P$ is said to be \emph{simple} if it is bounded by a closed
path, and every vertices are
distinct and edges intersect only at common endpoints.  The polygon
$P$ is \emph{weakly simple} if, for any $\eps > 0$, there is a simple
polygon $Q$ such that the Fr\'{e}chet distance between $P$ and $Q$ is
at most $\eps$~\cite{weakly}.  
The algorithms we use in this paper are designed for simple
polygons, but they also work for weakly simple polygons.

The vertices of a simple polygon $P$ with $n$ vertices are labeled
$v_1, \ldots, v_n$ in clockwise order along the boundary of $P$.  We
set $v_{n+k}=v_k$ for all $k \geq 1$.  An edge whose endpoints are
$v_i$ and $v_{i+1}$ is denoted by $e_i$.  For ease of presentation, we
make the following \emph{general position assumption}: no vertex of
$P$ is equidistant from two distinct vertices of $P$, which was also
assumed in~\cite{FVD}.  This assumption can be removed by applying
perturbation to the degenerate vertices~\cite{perturb}.

For any two points $x$ and $y$ lying inside a (weakly) simple polygon
$P$, the \emph{geodesic path} between $x$ and $y$, denoted by
$\pi(x,y)$, is the shortest path inside $P$ between $x$ and $y$.  The
length of $\pi(x,y)$ is called the \emph{geodesic distance} between
$x$ and $y$, denoted by $d(x,y)$.  The geodesic path between any two
points in $P$ is unique.  The geodesic distance and the geodesic path
between $x$ and $y$ can be computed in $O(\log n)$ and $O(\log n + k)$
time, respectively, after an $O(n)$-time preprocessing, where $k$ is the
number of vertices on the geodesic path~\cite{shortest-path-tree}.
The vertices of $\pi(x,y)$ excluding $x$ and $y$ are reflex
vertices of $P$ and they are called the \emph{anchors} of $\pi(x,y)$.
If $\pi(x,y)$ is a line segment, it has no anchor.   In this paper,
``distance'' refers to geodesic distance unless specified otherwise.
 
Given a set $X$ of points in $P$ (for instance a polygon or a disk), we use $\partial X$
to denote the boundary of $X$.  
A set $X \subseteq P$ is \emph{geodesically convex} if
$\pi(x,y) \subset X$ for any two points $x$ and $y$ in $X$.
For any two points $u$ and $w$ on $\partial
P$, let $\subchain{u}{w}$ be the part of $\partial P$ in clockwise
order from $u$ to $w$.  
For $u=w$, let $\subchain{u}{w}$ be the vertex $u$.
The subpolygon of $P$ bounded by
$\subchain{u}{w}$ and $\pi(u,w)$ is denoted by $\subpolygon{u}{w}$.
Note that $\subpolygon{u}{w}$ may not be simple, but it is always
weakly simple.  Indeed, consider the set of Euclidean disks centered
at points on $\pi(u,w)$ with radius $\eps > 0$.  There exists a simple polygonal
curve connecting $u$ and $w$ that lies in the union of these disks and
that does not intersect $\subchain{u}{w}$ except at $u$ and $w$.  
The region bounded by that
simple curve and $\subchain{u}{w}$ is a simple polygon whose
Fr\'{e}chet distance from $P$ is at most $\eps$.  

The \emph{radius} of
$P$, denoted by $r(P)$, is defined as $\max_{p \in P}d(c,p)$, where
$c$ is the geodesic $1$-center of $P$.  Given two points $\alpha,
\beta \in \partial P$, we set $\radbd{\alpha}{\beta} =
\rad{\subpolygon{\alpha}{\beta}}$.
Notice that $\radbd{\alpha}{x}$ is monotonically increasing as $x$ moves clockwise from $\alpha$ along $\partial P$. Similarly, $\radbd{x}{\alpha}$ is monotonically decreasing as $x$ moves clockwise from $\alpha$ along $\partial P$.

The geodesic disk centered at a point $p \in P$ with radius $r$,
denoted by $D_r(p)$, is the set of points whose geodesic distances
from $p$ are at most $r$.  The boundary of a geodesic disk inside $P$
consists of disjoint polygonal chains of $\partial P$ and $O(n)$
circular arcs~\cite{geodesic-disk}.  Given a center $p \in P$ and a
radius $r \in \mathbb{R}$, $D_r(p)$ can be computed in $O(n)$ time as
follows.  We first compute the shortest path map of $p$ in linear
time~\cite{shortest-path-tree}.  Each cell in the shortest
path map of $p$ is a triangle and every point $q$ in the same cell has the
same combinatorial structure of $\pi(p,q)$.  Thus, a cell in the
shortest path map of $p$ intersects at most two circular arcs of
$D_r(p)$.  Moreover, a circular arc intersecting a cell $C$ is a part
of the boundary of the Euclidean disk centered at $v$ with radius
$r-d(p,v)$, where $v$ is the (common) anchor of $\pi(p,q)$ closest to $q$ for
a point $q \in C$, if it exists, or $p$ itself, otherwise.  With this fact, we can compute $\partial D_r(p)$
by traversing the cells from a cell to its neighboring cell
 and computing the circular arcs of $\partial
D_r(p)$ in time linear in the number of cells and circular
arcs, which is $O(n)$.  

We call a set of two points $c_1,c_2 \in P$ a \emph{$2$-set}.  For instance,
a geodesic $2$-center of $P$ is a $2$-set.  We slightly abuse notation and
write $(c_1,c_2)$ (instead of the usual notation $\{c_1,c_2\}$ for a set) to
designate the $2$-set defined by $c_1$ and $c_2$.  The radius of a
$2$-set $(c_1,c_2)$ in $P$, denoted by $r_P(c_1,c_2)$, is defined as
\begin{linenomath}
$$r_P(c_1,c_2) = \max_{p \in P} \min\{d(c_1,p),d(c_2,p)\} .$$
\end{linenomath}
A geodesic $2$-center of $P$ is a $2$-set with minimum radius.  Note
that given any $2$-set $(c_1,c_2)$ and $r \geq r_P(c_1,c_2)$, 
it holds that $P\subseteq D_r(c_1)\cup D_r(c_2)$.

For any two points $x,y \in P$, the \emph{bisector} of $x$ and $y$ is
defined as the set of points in $P$ equidistant from $x$ and $y$.  The
bisector of two points may contain a two-dimensional region if there
is a vertex of $P$ equidistant from $x$ and $y$.
If we remove all two-dimensional regions from the bisector,
we are left with curves each of which is contained in $P$ with two
endpoints on $\partial P$.  Among such curves, we call the one
crossing $\pi(x,y)$ the \emph{bisecting curve} of $x$ and $y$,
denoted by $b(x,y)$.  See Figure~\ref{fig:bisector}(a).

\begin{figure}[t]
  \begin{center}
    \includegraphics[width=0.7\textwidth]{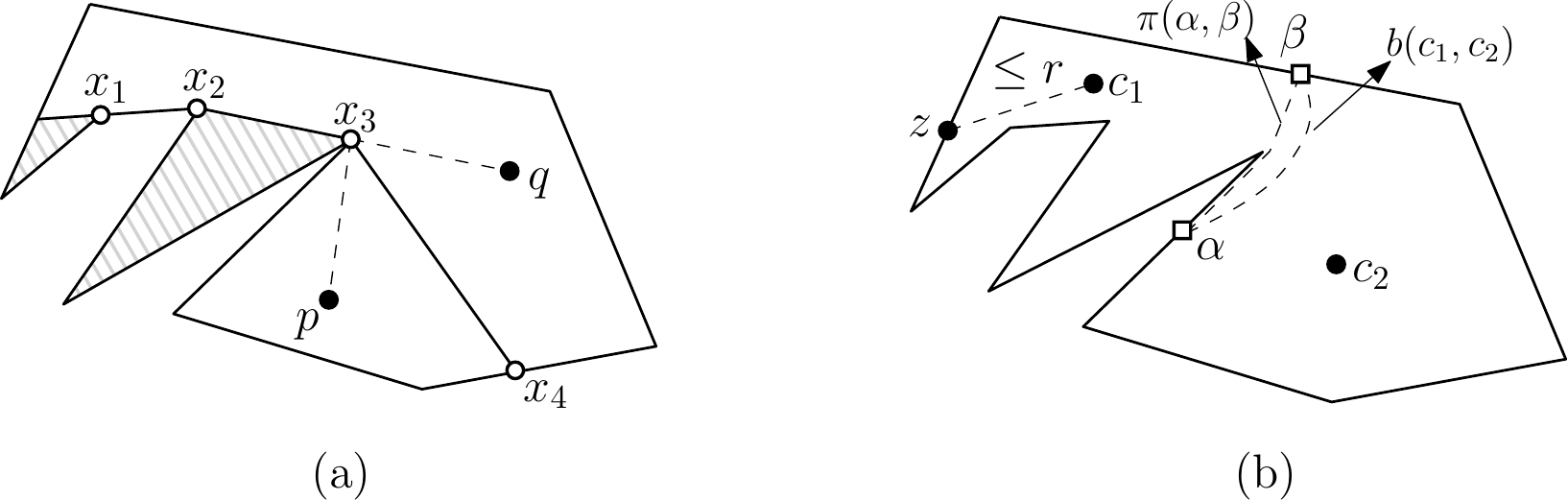}
    \caption{\small (a) The bisector of $p$ and $q$ contains
        two 2-dimensional regions (two dashed triangles). We have two
        line segments $x_1x_2$ and $x_3x_4$ after removing two
        2-dimensional regions. Here, $b(p,q)=x_3x_4$.  (b) For any
        point $z$ in $\subchain{\alpha}{\beta}$, we have $d(z,c_1)
        \leq r$ for any radius $r$ for the geodesic disks containing $P$ centered
        at $c_1$ and $c_2$, where $\alpha$ and $\beta$ are the endpoints of $b(c_1,c_2)$.}
    \label{fig:bisector}
  \end{center}
\end{figure}

\section{The partition by a \texorpdfstring{$2$}{2}-center}
Although there may exist more than one geodesic $2$-center of $P$, the radius
of any geodesic $2$-center is the same.  Let $(c_1^*, c_2^*)$ be a
geodesic $2$-center and $r^* = r_P(c_1^*,c_2^*)$.  For any two points
$\pu$ and $\pw$ on $\partial P$, let $\maxrad{\pu}{\pw} =
\max\{\radbd{\pu}{\pw},\radbd{\pw}{\pu}\}$.
We say that two geodesic disks \emph{cover} $P$ if the union of the two geodesic
disks coincides with $P$.

\begin{lemma}\textup{\cite[Lemma 1]{1-center1989}}
  \label{path-convex}
  Let $a, b$ and $c$ be points in $P$.  As $x$ varies along
  $\pi(b,c)$, $d(a,x)$ is a convex function of $d(b,x)$, and $d(a,x)
  \leq \max\{d(a,b),d(a,c)\}$.
\end{lemma}

\begin{lemma} \label{partition} If $P$ is covered by two geodesic
  disks centered at points in $P$ with radius $r$, then there are two
  points $x, y \in \partial P$ with $\maxrad{x}{y} \leq r$.
\end{lemma}
\begin{proof}
  Let $c_1$ and $c_2$ be the centers of the two geodesic disks with
  radius $r$ covering $P$.  Let $\pu$ and $\pw$ be the
  two endpoints of the bisecting curve $b(c_1,c_2)$.  We will argue that
  $\maxrad{\pu}{\pw} \leq r$.  
  
  Without loss of generality, assume that
  $c_1$ lies in the subpolygon of $P$ bounded by $b(c_1,c_2)$ and
  $\subchain{\pu}{\pw}$.  Let $z$ be any point on
  $\subchain{\pu}{\pw}$.  See Figure~\ref{fig:bisector}(b).
  Since $P$ coincides with $D_r(c_1)\cup D_r(c_2)$, we have
  $\min\{d(z,c_1),d(z,c_2)\} \leq r$.  Also, since $z$ and $c_1$ lie
  in the same side of $b(c_1,c_2)$, we have $d(z,c_1) \leq d(z,c_2)$.
  Therefore, $d(z,c_1) = \min\{d(z,c_1),d(z,c_2)\} \leq r$.  
  
  Moreover,
  for any point $p \in \pi(\pu,\pw)$, it holds that $d(c_1,p) \leq
  \max\{d(c_1,\pu),d(c_1,\pw)\}$ by Lemma~\ref{path-convex}.  Then,
  since $\pu$ and $\pw$ are the endpoints of $b(c_1,c_2)$, we find
  $\max\{d(c_1,\pu),d(c_1,\pw)\} \leq r$, from which $d(c_1,p) \leq
  r$.  Therefore, the boundary of $\subpolygon{\pu}{\pw}$ is contained in
  $D_r(c_1)$ and so is $\subpolygon{\pu}{\pw}$ by the geodesic
  convexity of $\subpolygon{\pu}{\pw}$.  
  
  Similarly, we can show that
  $\subpolygon{\pw}{\pu}$ is contained in $D_r(c_2)$.  Consequently,
  $\maxrad{\pu}{\pw} \leq r$.
\end{proof}

For any $2$-set $(c_1,c_2)$ in $P$ and any radius $r$, we call a pair $(\alpha,\beta)$
of points on $\partial P$ a \emph{point-partition of $P$
  with respect to $(c_1,c_2,r)$} if $d(c_1,x) \leq r$ and
$d(c_2,y) \leq r$ for all points $x \in \subpolygon{\alpha}{\beta}$
and $y \in \subpolygon{\beta}{\alpha}$.  Note that a point-partition
with respect to $(c_1,c_2,r)$ does not exist if $r < r_P(c_1,c_2)$.  A
pair $(e,e')$ of edges is called an \emph{edge-partition with respect
  to $(c_1,c_2, r)$} if there is a point-partition
$(\alpha,\beta)$ with respect to $(c_1,c_2,r)$ for $\alpha \in e$
and $\beta \in e'$.  A point-partition and an edge-partition with
respect to $(c_1^*, c_2^*, r^*)$ are said to be \emph{optimal}.  By
Lemma~\ref{partition}, there always exist an optimal point-partition
and an optimal edge-partition in a simple polygon.  Note that a
point-partition and an edge-partition with respect to $(c_1,c_2,r)$
are not necessarily unique if $\min\{d(c_1,\alpha), d(c_1,\beta)\}<r$,
where $\alpha$ and $\beta$ are the two endpoints of $b(c_1,c_2)$.  If
an optimal point-partition $(\pu,\pw)$ of $P$ is given, we can compute
a $2$-center in linear time using the algorithm in~\cite{1-center-SoCG,1-center}.

Our general strategy is to first compute a set of pairs of edges,
which we call \emph{candidate edge pairs}, containing at least one
optimal edge-partition.  For each candidate edge pair $(e_i,e_j)$, we
compute a \emph{$2$-center $(c_1,c_2)$ restricted to $(e_i,e_j)$}.
That is, a $2$-set $(c_1,c_2)$ such that $c_1$ and $c_2$ are the
$1$-centers of $\subpolygon{\alpha}{\beta}$ and
$\subpolygon{\beta}{\alpha}$, respectively, where $(\alpha, \beta)$ is
the pair realizing $\inf_{(x,y) \in e_i \times e_j} \maxrad{x}{y}$.

\subsection{Computing a set of candidate edge pairs}

In this section, we define \emph{candidate edge pairs} and describe
how to find the set of all candidate edge pairs in $O(n^2\log n)$ time.  Let
$f(\cdot)$ be the function which maps each vertex $v$ of $P$ to
the set of vertices $v'$ of $P$ that minimize
$\maxrad{v}{v'}$.
It is possible that there are more than one vertex $v'$ that minimizes
$\maxrad{v}{v'}$. Moreover, such vertices 
appear on the boundary of $P$ consecutively.
This is because
the function $\radbd{v}{x}$ is non-decreasing and $\radbd{x}{v}$
is non-increasing as $x$ moves clockwise from $v$ along
$\partial P$.

We use $f_{cw}(v)$ to denote the set of all vertices on $\bd P$ that come
after $v$ and before any vertex in $f(v)$ in clockwise order.
Similarly, we use $f_{ccw}(v)$ to denote the set of all vertices on $\bd P$ that 
come after $v$ and before any vertex in $f(v)$ in counterclockwise order.
Refer to Figure~\ref{fig:candidate_pair}.
The three sets $f_{ccw}(v)$, $f(v)$ and $f_{cw}(v)$ are pairwise disjoint
by the fact that $v \notin f(v)$ and 
by the monotonicity of  $\radbd{v}{x}$ and  $\radbd{x}{v}$.

\begin{figure}[t]
  \begin{center}
    \includegraphics[width=0.5\textwidth]{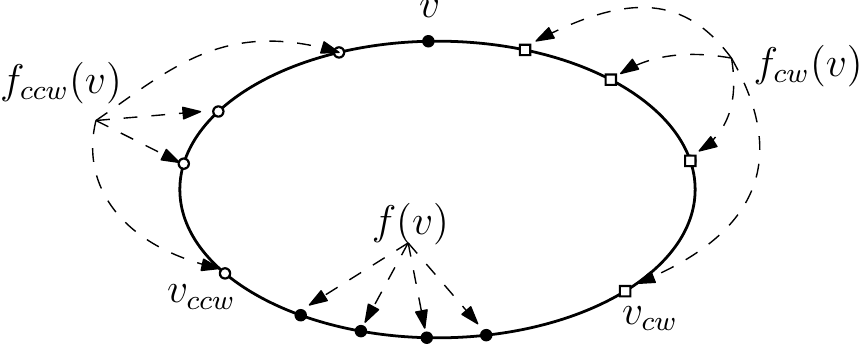}
    \caption{\small  For a vertex $v$, we define $f_{ccw}(v)$, $f_{cw}(v)$ and
    	$f(v)$.}
    \label{fig:candidate_pair}
  \end{center}
\end{figure}

Given two points $\alpha, \beta \in \partial P$, recall that we set
$\radbd{\alpha}{\beta} = \rad{\subpolygon{\alpha}{\beta}}$.
\begin{lemma}
  \label{convex radius-neighbor}
  Given a vertex $v$ of $P$, it holds that
  $\radbd{v}{w} < \radbd {w}{v}$ for any vertex $w \in f_{cw}(v)$
  and $\radbd{v}{w} > \radbd{w}{v}$ for any vertex $w \in f_{ccw}(v)$.
\end{lemma}
\begin{proof}
  Let us focus on the first inequality.
  Assume to the contrary that $\radbd{v}{w} \geq
  \radbd{w}{v}$ for some vertex $w \in f_{cw}(v)$.
  Let $v'$ be a vertex in $f(v)$.
  Since
  $\radbd{v}{v'} \geq \radbd{v}{w}$ and $\radbd{w}{v} \geq
  \radbd{v'}{v}$, we have $\radbd{v}{v'} \geq \radbd{v'}{v}$.
  Thus $\maxrad{v}{v'} = \radbd{v}{v'} \geq \radbd{v}{w} =
  \maxrad{v}{w}$, which contradicts the fact that $f(v)\cap f_{cw}(v)=\phi$.
  
  We can prove the second inequality in a similar way.
\end{proof}

\begin{lemma}
	\label{ordering}
	Let $v$ be a vertex of $P$.
	For a vertex $w \in f_{cw}(v)$, it holds that
	$f(w) \cap \subchain{v_{cw}}{w} \neq \phi$, where
	$v_{cw}$ is the last vertex of $f_{cw}(v)$ from $v$ in clockwise order.
\end{lemma}
\begin{proof}
  Let $w'$ be the last vertex of $f(w)$ from $w$ in clockwise order.
  We show that $w' \in \subchain{v_{cw}}{w}$, 
  which implies the lemma.
  
  Assume to the contrary that $w' \in \subchain{w}{v_{cw}}\setminus\{w,v_{cw}\}$.
  By Lemma~\ref{convex radius-neighbor} and the monotonicity of
  the functions $\radbd{\cdot}{w'}$ and $\radbd{w'}{\cdot}$, we have
  $\radbd{w}{w'} \leq \radbd{v}{w'} < \radbd{w'}{v} \leq
  \radbd{w'}{w}$.  Similarly, we get $\radbd{w}{v_{cw}} \leq
  \radbd{v}{v_{cw}} < \radbd{v_{cw}}{v} \leq \radbd{v_{cw}}{w}$.
  Thus, $\maxrad{w}{v_{cw}} = \radbd{v_{cw}}{w} \leq \radbd{w'}{w} =
  \maxrad{w}{w'}$ by the monotonicity of $\radbd{\cdot}{w}$.  This
  contradicts the fact that $w'$ is the last vertex of $f(w)$ from $w$ in clockwise order.
\end{proof}

\begin{figure}
  \begin{center}
    \includegraphics[width=0.7\textwidth]{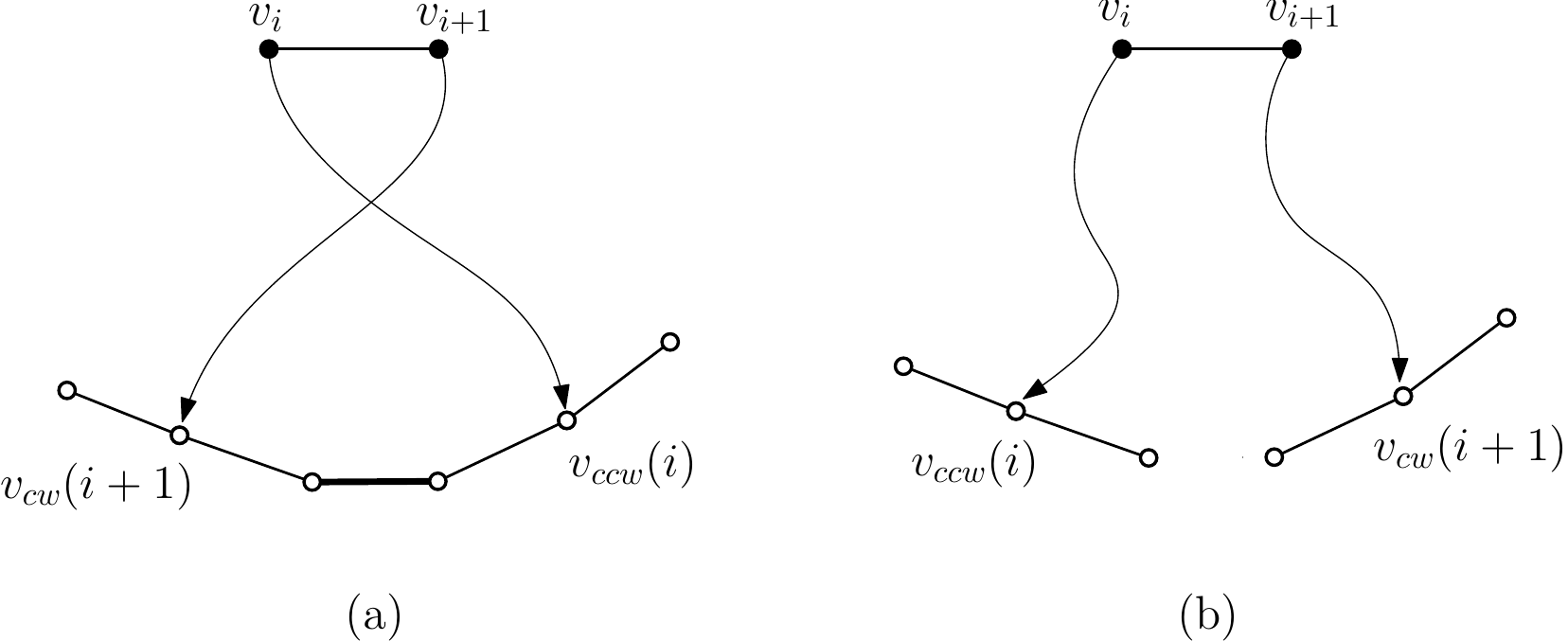}
    \caption{\small According to the relative positions for
        $v_{cw}(i+1)$ and $v_{ccw}(i)$, the candidate edges of $e_i$ are
        defined.  (a) Four candidate edges of type (1) and one candidate edge of type (2).
        (b) Four candidate edges of type (1).}
    \label{fig:three_type_candidate}
  \end{center}
\end{figure}

\label{procedure-candidate edge pairs}
Given an edge $e_i=v_iv_{i+1}$, an edge $e_j=v_jv_{j+1}$ is called a
\emph{candidate edge of $e_i$} if it belongs one of the following
two types. 
For a vertex $v_k$, 
let $\vcw{k}$ be the last vertex of $f_{cw}(v_k)$ from $v_k$ in clockwise order
and $\vccw{k}$ be the first vertex of $f_{ccw}(v_{k})$ from $v_{k}$
in counterclockwise order.
Refer to Figure~\ref{fig:three_type_candidate}.
\begin{enumerate}
\item $v_j$ (or $v_{j+1}$) is $\vccw{i}$ or $\vcw{i+1}$.
\item $e_j$ has both endpoints in the interior of
  $\subchain{\vccw{i}}{\vcw{i+1}}$ and $\vccw{i}$ comes before $\vcw{i+1}$ from $v_i$ 
  in clockwise order along $\bd P$.
  The edge marked with thick line segment in
    Figure~\ref{fig:three_type_candidate}(a) satisfies this condition.
\end{enumerate}
A pair $(e_i, e_j)$ of edges is called a \emph{candidate edge pair} if
$e_j$ is a candidate edge of $e_i$.
\begin{lemma}
  \label{optimal partition}
  There is an optimal edge-partition in the set of all candidate edge
  pairs.
\end{lemma}
\begin{proof}
By Lemma~\ref{partition}, there exists an optimal point-partition. Among all optimal point-partitions, let $(\pu,\pw)$ be one such that $(\pu,\pw')$ is not an optimal point-partition
for any $\pw' \in \subchain{\pu}{\pw}\setminus \{\beta\}$.
Let $(e_i,e_j)$ be an optimal edge-partition
with $\pu \in e_i$ and $\pw \in e_j$.
If $\pw$ is a vertex,
let $e_j$ be an edge such that $\pw = v_{j+1}$
(so that in all cases,
the counterclockwise neighbor of $\pw$ is $v_j$).
Our goal is to show that if 
there is no candidate edge pair of type (1),
then $(e_i,e_j)$ is a candidate edge pair of type (2).
Thus,
we need to locate $e_j$ with respect to $\vccw{i}$ and $\vcw{i+1}$.

Assume that $(e_i,e_j)$ is not a candidate edge pair of type (1).
We claim the followings.
\begin{itemize}
	\item[1.] $\vccw{i}\in\subchain{v_{i+1}}{v_{j+1}}$
	\item[2.] $\vcw{i+1}\in\subchain{v_j}{v_i}$ 
\end{itemize} 
Suppose these two claims are true.
Then, $\vcw{i+1}$ appears after $\vccw{i}$ as
we move clockwise from $v_i$ along $\bd P$ since we assume that
$(e_i,e_j)$ is not a candidate edge pair of type (1).
Moreover,
$e_j$ has both endpoints in the interior of $\subchain{\vccw{i}}{\vcw{i+1}}$.
Therefore,
$(e_i, e_j)$ is a candidate edge pair of type (2).

It remains to prove the two claims.
We start with the first one.
The strategy is to show that if the first claim is not true,
there is another optimal edge-partition belonging to type (1).
Let $x_i'\in\partial P$ be the last clockwise \emph{point}
from $v_i$ which minimizes
$\maxrad{v_i}{x_i'}$.
By the definitions of $x_i'$ and $(\pu,\pw)$,
we have
\begin{alignat*}{3}
\maxrad{v_i}{x_i'} &=r(v_i,x_i') &&= r(x_i',v_i) ,\\
\maxrad{\pu}{\pw} &=r(\pu,\pw) &&= r(\pw,\pu) .
\end{alignat*}
Since $(\pu,\pw)$ is an optimal point-partition,
we have
\begin{align}
\label{ineq Lemma 6}
r(x_i',v_i) = \maxrad{v_i}{x_i'} \geq \maxrad{\pu}{\pw} = r(\pw,\pu) .
\end{align}
If our first claim is not true,
then $\vccw{i}$ comes after $v_{j+1}$
in clockwise order from $v_i$.
This implies that
$x_i'$ comes after $v_{j+1}$ 
in clockwise order from $v_i$.
We show that $\maxrad{v_i}{x_i'} = \maxrad{\pu}{\pw}$.
If not, we would have $r(x_i',v_i) > r(\pw,\pu)$,
from which,
by the monotonicity of the functions $r(\pw,\cdot)$ and $r(\cdot,v_i)$,
$$r(x_i',v_i) > r(\pw,\pu) \geq r(\pw,v_i) \geq r(x_i',v_i),$$
which is a contradiction.
Therefore,
$r(x_i',v_i) = r(\pw,\pu)$, which means that $(v_i,x_i')$ is an optimal point-partition since we now have
$\maxrad{v_i}{x_i'} = r(x_i',v_i)
=  r(\pw,\pu) = \maxrad{\pu}{\pw}$.

Let us redefine $\pu$ as $v_i$ and $\pw$ as $x_i'$.
We also redefine $(e_i,e_j)$ as an optimal edge-partition
such that $v_i = \pu$ and $\pw \in e_j$
(if $\pw$ is a vertex, we choose $e_j$ such that $\pw = v_{j+1}$).
In this way, $\vccw{i}$ remains the same.
Thus $\vccw{i}$ is the counterclockwise neighbor of $x_i'$.
Therefore, $(e_i, e_j)$ is a candidate pair of type (1).

We now prove the second claim. The second claim can be proved in a similar way.
Assume to the contrary that $\vcw{i+1} \in \subchain{v_i}{v_j}\setminus\{v_i,v_j\}$.
Then $\vcw{i+1}$ comes before $v_j$ from $v_{i+1}$ in clockwise order.
Let $x_{i+1}'$ be the first clockwise \emph{point} from $v_{i+1}$ that minimizes $\maxrad{v_{i+1}}{x_{i+1}'}$.
Thus, $x_{i+1}'$ comes before $\beta$ from $v_{i+1}$ in clockwise order.
Then the following holds:
$$\maxrad{v_{i+1}}{x_{i+1}'} = r(v_{i+1},x_{i+1}')\leq r(v_{i+1},\beta) 
\leq r(\alpha,\beta)=\maxrad{\alpha}{\beta}.
$$
This implies that $(v_{i+1},x_{i+1}')$ is also an optimal point-partition.
We redefine $\alpha$ as $v_{i+1}$ and $\beta$ as $x_{i+1}'$.
We also redefine $e_i$ and $e_j$ accordingly.
This pair $(e_i,e_j)$ is a candidate edge pair of type (1).

Therefore, we have an optimal edge-partition in the set of 
all candidate edge pairs.
\end{proof}

\begin{figure}
  \begin{center}
    \includegraphics[width=0.45\textwidth]{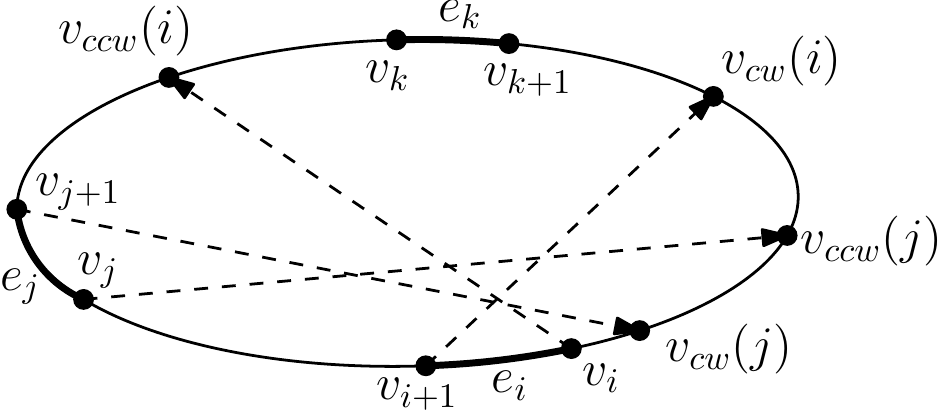}
    \caption{\small Two candidate edge pairs $(e_i,e_k)$ and $(e_j,e_k)$ of type (2).
    	$e_i$ comes before $e_j$ from $e_k$.}
    \label{fig:candidate_proof}
  \end{center}
\end{figure}

\begin{lemma}
  \label{the numbero f candidate edge pairs}
  There are $O(n)$ candidate edge pairs.
\end{lemma}

\begin{proof}
  Since $\vccw{i}$ and $\vcw{i}$ are uniquely defined for any vertex $v_i$ of $P$, the total number of candidate edge pairs of type (1) is at most $4n$.
  
  Now we consider the candidate edge pairs which have not been counted
  yet.  Assume that for an edge $e_k$ there are two distinct candidate edge
  pairs, say $(e_i,e_k)$ and $(e_j,e_k)$, of type (2).  Without loss
  of generality, we assume that $e_i$ comes before than $e_j$ in
  clockwise order from $e_k$.  
  Since they are candidate edge pairs of type (2), $e_k$ is contained
  in the intersection of $\subchain{\vccw{i}}{\vcw{i+1}}$ and
  $\subchain{\vccw{j}}{\vcw{j+1}}$.
  
 We now argue that $v_j$ lies on $\subchain{v_{i+1}}{\vcw{i+1}}$. Suppose that $v_j \in \subchain{\vcw{i+1}}{v_{i+1}}\setminus\{\vcw{i+1},v_{i+1}\}$ for the sake of a contradiction. Then, since $e_k$ is contained
 in $\subchain{\vccw{i}}{\vcw{i+1}}$, the vertex $v_j$ lies in 
 the interior of $\subchain{v_{k+1}}{v_{i+1}}$. This contradicts the fact that $e_i$ comes before than $e_j$ in
 clockwise order from $e_k$. Therefore, $v_j$ lies on $\subchain{v_{i+1}}{\vcw{i+1}}$, which means
 that $v_j \in f_{cw}(v_{i+1})$. 
 Refer to Figure~\ref{fig:candidate_proof}.
  
  Consequently, by Lemma~\ref{ordering},
  $\vccw{j}$ lies in $\subchain{\vcw{i+1}}{v_j}$. Since
  $e_k$ is contained in $\subchain{\vccw{i}}{\vcw{i+1}}$, 
  $\vcw{i+1}$ lies in $\subchain{v_{k+1}}{v_i}$.
  This implies that $e_k$ is not contained in $\subchain{\vccw{j}}{\vcw{j+1}}$, which is
  a contradiction.
  
  Therefore, for an edge $e$, there exists at most one edge $e'$ such
  that $(e', e)$ is a candidate edge pair of type (2).  Thus the
  number of candidate edge pairs of type (2) is $O(n)$.
\end{proof}

Now we present a procedure for finding the set of all candidate edge
pairs.  
For each index $i$, we compute $\vcw{i}$ and $\vccw{i}$ in $O(n \log n)$ time.
To find $\vcw{i}$, we apply binary search on the vertices of $P$
using Lemma~\ref{convex radius-neighbor} and a linear-time algorithm~\cite{1-center-SoCG,1-center} 
that computes the center of a simple polygon.
We can find $\vccw{i}$ in a similar way.
This takes $O(n^2\log n)$ time in total. 

Then, we compute the set of all candidate edge pairs based on
the information we just computed. For each edge $e_i$, we consider the
edges lying between $\vccw{i}$ and $\vcw{i+1}$ if $\vccw{i}$ comes before $\vcw{i+1}$
from $v_i$ in clockwise order. Otherwise, we consider the four edges
incident to $\vccw{i}$ and $\vcw{i+1}$.  In total, this
takes time linear to the number of candidate edge pairs, which
is $O(n)$ by Lemma~\ref{the numbero f candidate edge pairs}.

\begin{lemma}
  The set of all candidate edge pairs can be computed in $O(n^2 \log n)$
  time.
\end{lemma}

\section{A decision algorithm for a candidate edge pair}
\label{sec:decision algorithm}

We say that a point-partition $(\pu,\pw)$ is \emph{restricted} to
$(e_i, e_j)$ if $\pu \in e_i$ and $\pw \in e_j$.  We say that a
triplet $(c_1,c_2,r)$ consisting of a $2$-set $(c_1,c_2)$ and a radius
$r$ is \emph{restricted} to $(e_i,e_j)$ if some point-partitions with
respect to $(c_1,c_2,r)$ are restricted to $(e_i,e_j)$.  We consider
$\maxrad{\alpha}{\beta}$ as a function whose variables are $\alpha
\in e_i$ and $\beta \in e_j$.  Since the function is continuous and
the domain is bounded, there exist two points, $\alpha^*\in e_i$ and
$\beta^*\in e_j$, that minimize the function.  We call $(c_1^*,c_2^*)$
\emph{a $2$-center restricted to $(e_i,e_j)$} if $c_1^*$ and $c_2^*$
are the $1$-centers of $\subpolygon{\alpha^*}{\beta^*}$ and
$\subpolygon{\beta^*}{\alpha^*}$, respectively.  By Lemma~\ref{optimal
  partition}, there is a $2$-center restricted to a candidate edge
pair 
which is a $2$-center (without any restriction).

In this section, we present a decision algorithm for a candidate edge
pair $(e_i,e_j)$.  Let $\radrestricted$ be the radius of a $2$-center
restricted to $(e_i,e_j)$. Let $r$ be an input of the algorithm.  The
decision algorithm in this section returns ``yes'' if $r \geq
\radrestricted$.  Additionally, it returns a $2$-center restricted to
$(e_i,e_j)$ with radius $r$.  It returns ``no'', otherwise.

Throughout this section, we assume that $\radbd{v_{i+1}}{v_j} \leq r
< \radbd{v_i}{v_{j+1}}$ and $\radbd{v_{j+1}}{v_i} \leq r <
\radbd{v_j}{v_{i+1}}$ because the other cases can be handled easily: if
$\radbd{v_{j+1}}{v_i} > r$ or $\radbd{v_{i+1}}{v_j} > r$, we return
``no''.  For the remaining cases, we return ``yes''.

The decision algorithm first assumes that $r \geq \radrestricted$ and
constructs a $2$-center restricted to $(e_i,e_j)$ with radius $r$.
The $2$-center produced by the algorithm is valid if and only if $r
\geq \radrestricted$.  Therefore, the algorithm can then decide
whether $r \geq \radrestricted$ by checking whether the $2$-center is
valid.  Thus, from now on, we assume that $r \geq \radrestricted$.  Let
$(c_1,c_2,r)$ be a triplet consisting of a $2$-set $(c_1,c_2)$ and
radius $r$ which is restricted to $(e_i,e_j)$, and $(\pu,\pw)$ be a
point-partition with respect to $(c_1,c_2, r)$ which is restricted to
$(e_i, e_j)$.  Without loss of generality, we assume that $D_{r}(c_1)$
contains $\subpolygon{\pu}{\pw}$ and $D_{r}(c_2)$ contains
$\subpolygon{\pw}{\pu}$.

\subsection{Computing the intersection of geodesic disks}
\label{sec:decision_intersection}
The first step of the decision algorithm is to compute the intersection
$I_1$ of the geodesic disks of radius $r$ centered at $v \in \subchain{v_{i+1}}{v_j}$
and the intersection $I_2$ of the geodesic disks of radius $r$ centered at 
$v \in \subchain{v_{j+1}}{v_i}$,
that is,  $I_1 =\cap_{k=i+1}^{j} {D_{r}(v_k)}$ and $I_2 = \cap_{k=j+1}^{i}
{D_{r}(v_k)}$.  Clearly, $c_1 \in I_1$ and $c_2 \in I_2$.

We compute $I_1$ and $I_2$ by constructing the farthest-point
geodesic Voronoi diagrams, denoted by $\mathsf{FV}_1$ and
$\mathsf{FV}_2$, of the vertices in $\subchain{v_{i+1}}{v_j}$ and the vertices in 
$\subchain{v_{j+1}}{v_i}$, respectively.  
For the case that the sites are on the vertices of 
$P$, the diagram can be computed in $O(n\log\log n)$
time~\cite{fvd_boundary}.  

A cell in $\mathsf{FV}_1$
associated with a site $t$ consists of the points $p \in P$ such that $t$ is the
site farthest from $p$ among all sites.  A \emph{refined} cell in
$\mathsf{FV}_1$ associated with site $t$ is obtained by further subdividing the
cell associated with site $t$ such that all points in the same refined cell have
the same combinatorial structure of the shortest paths from their common 
farthest site.  While constructing $\mathsf{FV}_1$ and
$\mathsf{FV}_2$, the algorithm~\cite{fvd_boundary} computes all refined cells.
For each refined cell, we can store the information of
the common farthest site $t$ of the refined cell and the anchor of $\pi(t,p)$ closest to $p$ for
a point $p$ in the refined cell.  

Then we compute circular arcs of $\partial I_1$ and $\partial I_2$
contained in a refined cell in time linear to the number of circular arcs in that refined cell plus the complexity of the refined cell.  By
traversing all refined cells, we can compute all circular arcs in $O(n)$ time
by the following lemma.

\begin{lemma}
  The total number of circular arcs in $\partial I_1$ and $\partial I_2$ is
  $O(n)$.
  \label{number of arcs}
\end{lemma}

\begin{proof}
  We prove the lemma only for $\partial I_1$. The case for $\partial
  I_2$ can be proven analogously.  The size of the (refined)
  farthest-point geodesic Voronoi diagram of $n$ sites in a simple
  polygon with $n$ vertices is $O(n)$~\cite{FVD}.  In other words,
  there are $O(n)$ \emph{refined} cells and edges of the Voronoi
  diagram.

  Let $s$ be a circular arc of $\partial I_1$. The center $c_s$ of the
  geodesic disk containing $s$ on its boundary lies in
  $\subchain{v_{i+1}}{v_j}$.  Note that $c_s$ is unique by the general
  position assumption.  Every geodesic disk whose center is a vertex
  in $\subchain{v_{i+1}}{v_j} \setminus \{c_s\}$ with radius $r$
  contains $s$ in its interior.  This means that, for any point $x \in
  s$, the farthest vertex from $x$ in $\subchain{v_{i+1}}{v_j}$ is
  $c_s$.  Moreover, the combinatorial structures of the geodesic paths
  from the center $c_s$ to points on the circular arc $s$ are the
  same.  Thus each circular arc $s$ is contained in a refined cell
  of $\mathsf{FV}_1$ whose site is $c_s$.
  Moreover, each endpoint of the circular arc lies in the boundary of
  the refined cell containing it (including the boundary of $P$.)
  Each edge of the diagram is intersected by at most one circular arc
  of $\partial I_1$.
  Therefore, the number of circular arcs in $\partial I_1$ is $O(n)$ 
  by the fact that the size
  of the refined farthest-point geodesic Voronoi diagram is $O(n)$.
\end{proof}

\begin{figure}
  \begin{center}
    \includegraphics[width=0.4\textwidth]{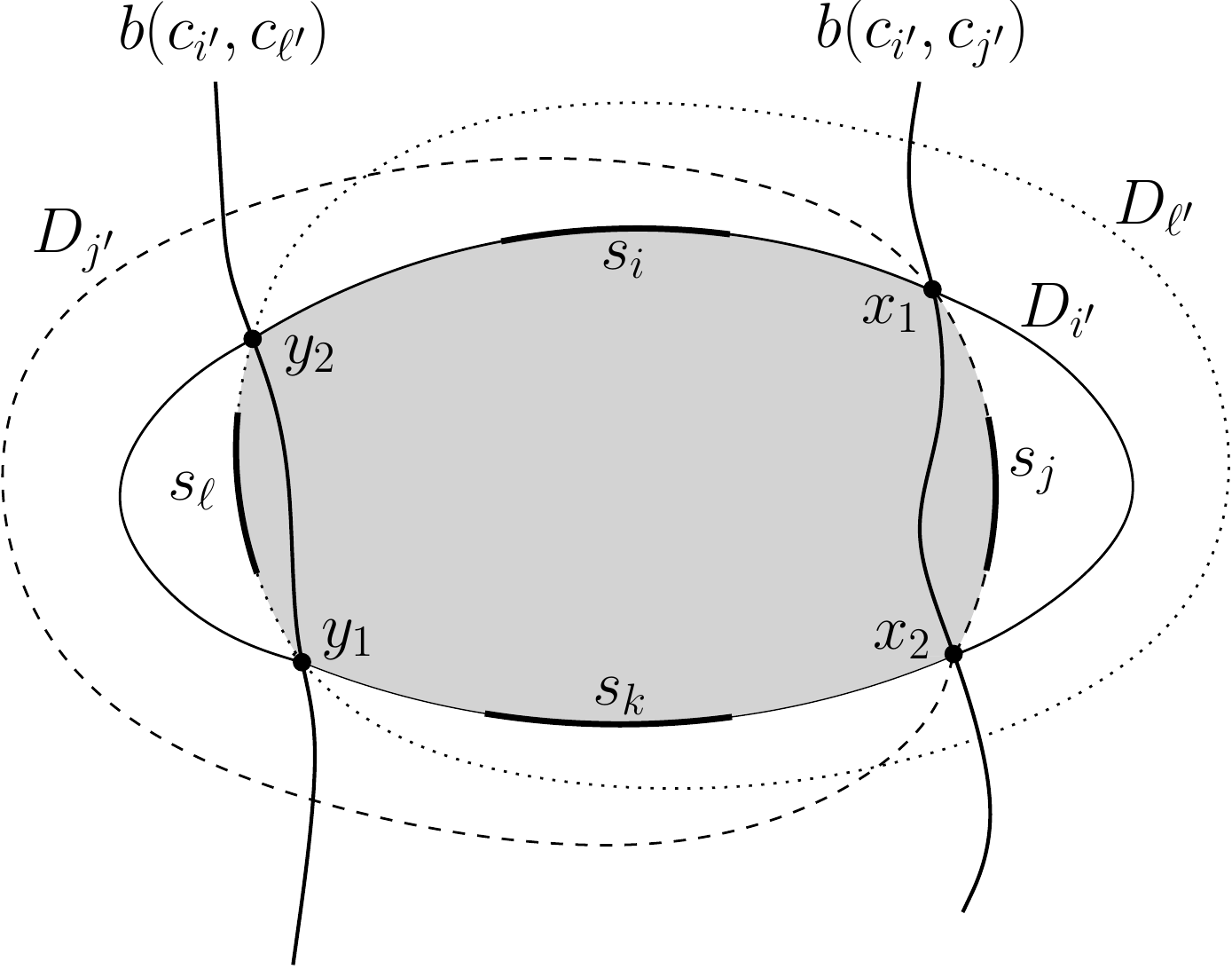}
    \caption{\small 
      The center of $D_{i'}$ must lie to either the
      right of $b(c_{i'},c_{j'})$  or
      the left of $b(c_{i'},c_{\ell'})$.
      \label{fig:geodesic_disk}}
  \end{center}
\end{figure}

\begin{lemma}
  \label{pseudo}
  Let $\mathcal {D}=\{D_1, \ldots, D_k \}$ be a set of geodesic disks
  with the same radius and let $I$ be the intersection of all geodesic disks in
  $\mathcal {D}$.  Let $S=\langle s_1, \ldots, s_{k'}\rangle$ be the cyclic sequence
  of the circular arcs of $\partial I$ along its boundary in clockwise
  order.  For any integer $i \in [1,k]$, the circular arcs in $\partial I
  \cap \partial D_i$ are consecutive in $S$.
\end{lemma}

\begin{proof}
  Assume to the contrary that there are four circular arcs $s_i, s_j,
  s_k, s_\ell$ in $S$ with $i < j< k < \ell$ such that $s_i, s_k
  \subset \partial D_{i'}$, $s_j \subset \partial D_{j'}$ and $s_\ell
  \subset \partial D_{\ell'}$ for three distinct geodesic disks
  $D_{i'}, D_{j'},D_{\ell'} \in \mathcal{D}$. See
  Figure~\ref{fig:geodesic_disk}. Let $c_{i'}, c_{j'}, c_{\ell'}$ be the
  centers of the disks $D_{i'}, D_{j'}, D_{\ell'}$, respectively. Then the
  bisecting curve of $c_{i'}$ and $c_{j'}$ intersects $\partial
  (D_{i'} \cup D_{j'})$ exactly twice.  Let $x_1$ and $x_2$ be these two
  intersection points such that $s_j$ is contained in the region
  bounded by $b(c_{i'},c_{j'})$ and the part of $\partial D_{i'}$ from
  $x_1$ to $x_2$ in clockwise order.  Similarly, the bisecting curve
  of $c_{i'}$ and $c_{\ell'}$ intersects $\partial (D_{i'} \cup
  D_{\ell'})$ exactly twice.  Let $y_1$ and $y_2$ be these intersection
  points such that $s_\ell$ is contained in the region bounded by
  $b(c_{i'},c_{\ell'})$ and the part of $\partial D_{i'}$ from $y_1$
  to $y_2$ in clockwise order.
  When we traverse $\bd I$ clockwise starting from $s_i$, we
  encounter $x_1$, $\interior{s_j}$, $x_2$, $\interior{s_k}$, $y_1$
  $\interior{s_\ell}$, and $y_2$ in order,
  where $\interior{s}$ is the circular arc $s$ excluding its endpoints
  for a circular arc $s$. 
  
  The center $c_{i'}$ lies in the subset $P_1 \subset D_{i'}$ bounded
  by $b(c_{i'},c_{j'})$ and the part of $\partial I$ from $x_1$ to
  $x_2$ in clockwise order.  On the other hand, $c_{i'}$ lies in the
  subset $P_2 \subset D_{i'}$ bounded by $b(c_{i'},c_{\ell'})$ and the
  part of $\partial I$ from $y_1$ to $y_2$ in clockwise order.  Thus,
  $c_{i'} \in P_1 \cap P_2$.
  Therefore, $P_1$ and $P_2$ must intersect. Since $c_{i'}$ lies in the interior of $D_{i'}$, $b(c_{i'},c_{\ell'})$ and $b(c_{i'},c_{j'})$ must intersect in the interior of $D_{i'}$.
  In order to satisfy the order of appearances of $x_1$, $x_2$, $y_1$ and $y_2$ along $\bd I$, $b(c_{i'},c_{\ell'})$ and $b(c_{i'},c_{j'})$ must intersect an even number of times in the interior of $D_{i'}$.
  This
   is impossible since $b(x,y)$ and $b(x,z)$ cross each other at most once
   for any three points $x, y, z$ in $P$.  Thus $P_1 \cap P_2 = \phi$, which
   is a contradiction.
\end{proof}

Note that $\partial I_1$ and $\partial I_2$ consist of $O(n)$ circular
arcs and (possibly incomplete) edges of $\partial P$ in total.  Let $\sarcs_1$
and $\sarcs_2$ be the unions of the circular arcs of $\partial I_1$
and $\partial I_2$, respectively.
By the following lemma, it is sufficient to choose two points, one
from $\sarcs_1$ and one from $\sarcs_2$, in order to find a $2$-center
restricted to $(e_i,e_j)$ with radius $r$.

\begin{lemma}	
  If $\radrestricted \leq r \leq
  \min\{\radbd{v_i}{v_{j+1}},\radbd{v_{j}}{v_{i+1}}\}$, there is a
  triplet $(c_1,c_2,r)$ restricted to $(e_i, e_j)$
  such that $c_1 \in \sarcs_1$ and $c_2 \in \sarcs_2$.
  \label{boundary is sufficient}
\end{lemma}

\begin{proof}
  Since $\radrestricted \leq r$, there is a triplet $(c_1', c_2', r)$
  made of a $2$-set $(c_1', c_2')$ and a radius $r$ restricted to
  $(e_i, e_j)$.  Let $(\pu, \pw)$ be a point-partition with respect to
  $(c_1',c_2',r)$ with $\pu \in e_i$ and $\pw \in e_j$.  Without loss
  of generality, we assume that $c_1' \in \subpolygon{\pu}{\pw}$ and
  $c_2' \in \subpolygon{\pw}{\pu}$.

  Consider $c_1'$ first.  Let $x \in e_i$ be the point closest
  to $v_i$ among the points satisfying $d(x,c_1')\leq r$.  Similarly, let $y \in e_j$ be
  the point closest to $v_{j+1}$ among the points satisfying $d(y,
  c_1')\leq r$. Then we
  have $\radbd{x}{y} \leq r$.  As we move $x$ from its current
  position to $v_{i}$ along $e_i$, $\radbd{x}{y}$ increases.  We move
  $x$ until $\radbd{x}{y} = r$ or $x$ reaches $v_{i}$.  If $x$ reaches
  $v_{i}$, we move $y$ from the current position to $v_{j+1}$ until
  $\radbd{x}{y} = r$.  This is always possible to find such $x$ and $y$ since, by the assumption, we have
  $\radbd{v_{i}}{v_{j+1}} \geq r$.

  Now we consider the subpolygon $\subpolygon{x}{y}$.  Let $c_1$ be
  the center of $\subpolygon{x}{y}$.  If there is a vertex $v \in
  \subchain{x}{y}$ with $d(v,c_1) = r$, then $c_1$ lies in $\sarcs_1$
  and $D_r(c_1)$ contains $D_r(c_1')$, thus we are done.  Otherwise,
  $c_1$ is the midpoint of $\pi(x,y)$.
  But then,
  since $D_r(c_2')$ contains $x$ and $y$,
  this means that $c_2' = c_1$, which is a contradiction.
  Thus, the pair
  $(c_1, c_2')$ is a $2$-center restricted to $(e_i, e_j)$ and $c_1 \in
  \sarcs_1$.
	
  Similarly, we can find $c_2$ lying in $\sarcs_2$ with $D_r(c_2')
  \subset D_r(c_2)$.
\end{proof}

\subsection{Subdividing the edges and the boundaries of the intersections}
\label{subsection Subdividing the edges and the chains}
The shortest path map rooted at $x$ is the subdivision of $P$
consisting of triangular cells such that every point $p$ in the same cell
has the same combinatorial structure of $\pi(p,x)$.  The map can be obtained
by extending the edges of the shortest path tree rooted at $x$ towards
their descendants~\cite{shortest-path-tree}.  Let $\spm{k}$ denote the
shortest path map rooted at $v_k$.  We compute the shortest path maps
$\spm{i}$ and $\spm{i+1}$.

By overlaying the two shortest path maps with $\partial I_1$, we
obtain the set of $O(n)$ \emph{finer} arcs of $\partial I_1$ as
follows.  We find any cell of $\spm{i}$ intersecting $\partial I_1$, and traverse to the
neighboring cells along $\partial I_1$.  Whenever we cross an edge of
the cell along an arc of $\bd I_1$, we compute the intersection
between the edge of the cell and the arc of $\partial I_1$. 
We can check in constant time whether a given arc of $\bd I_1$ crosses
an edge of a given cell in $\spm{i}$ since every cell is a triangle.
While traversing $\partial I_1$, we cross each edge of $\spm{i}$ at most
twice by the geodesic convexity of $I_1$.  Thus, in total, it is
sufficient to traverse $\partial I_1$ once and cross each edge in
$\spm{i}$ at most twice.  Similarly, we compute the intersections
between $\partial I_1$ and $\spm{i+1}$.  From now on, we treat
$\partial I_1$ as the sequence of $O(n)$ finer arcs.

\begin{figure}
  \begin{center}
    \includegraphics[width=0.8\textwidth]{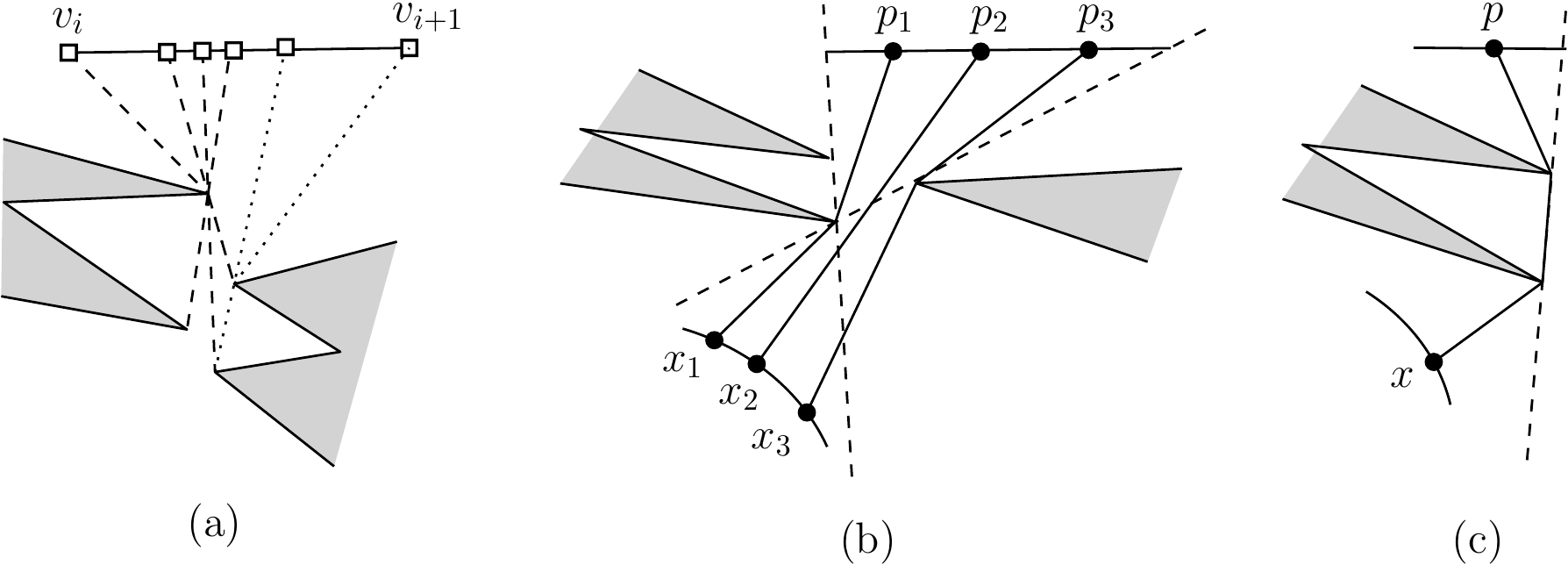}
		\caption{\small 
			(a) By extending the edges of the shortest path trees rooted at $v_i$
			and rooted at $v_{i+1}$, we obtain subedges of $e_i=v_iv_{i+1}$.
			The endpoints of the subedges are marked with squares.
			(b) If the geodesic path $\pi(x,p)$ is elementary,
			there are at most three different combinatorial structures of 
			$\pi(x,p)$ depending on $x$ and $p$.
			(c) For any point $x$ in the 
			finer circular arc and any point $p$ in the subedge,
			the combinatorial structures of $\pi(x,p)$ are the same for all $x$ and $p$.		
			\label{fig:combinatorial_structure}}
                    \end{center}
                  \end{figure}
	
We also subdivide the polygon edge $e_i$ into $O(n)$
\emph{subedges} by overlaying the extensions of the
edges in the shortest path trees rooted at $v_i$ and
$v_{i+1}$ towards their parents with $e_i$.  See
Figure~\ref{fig:combinatorial_structure}(a).
Let $\mathcal{L}_i$ be the set of intersections of
the extensions of the edges in the shortest path
trees of $v_i$ and $v_{i+1}$ with $e_i$.  While computing $\mathcal{L}_i$,
we sort them along $e_i$ from $v_{i+1}$. This takes $O(n\log n)$ time. 
We compute $\mathcal{L}_j$ similarly, which is the set of the
intersections of the extensions of the edges in
$\spm{j}$ and $\spm{j+1}$ with $e_j$,
and sort them along $e_j$ from $v_j$.

We say that a geodesic path between two points is \emph{elementary} if the
number of line segments in the geodesic path is at
most two.  
If $\pi(x,p)$ is elementary for all points $x$ on the same finer arc of $\bd I_1$ and all points $p$
on the same subedge,
there are at most three possible distinct combinatorial structures as shown in
Figure~\ref{fig:combinatorial_structure}(b).
However, the combinatorial structures of $\pi(x, p)$ are the same for any point $x$ on the
same finer circular arc and any point $p$ on the same subedge if $\pi(x,p)$ is not elementary for all
$x$ and $p$.  Refer to Figure~\ref{fig:combinatorial_structure}(c).

\subsection{Four coverage functions and their extrema}
\label{sec:coverage-functions}
In this section,
we will subdivide $\partial I_1$ and $\partial
I_2$ into $O(1)$ subchains (refer to Subsection~\ref{subsubsect subdividing}).
Then for every pair of
subchains, one from $\partial I_1$ and one
from $\partial I_2$, we will explain how to decide whether there is a
$2$-center $(c_1, c_2)$ restricted to a candidate edge pair
lying on the two subchains in Section~\ref{sec:computing a $2$-center for a pair of subchains}.  To this end, in the following subsection,
we
define four functions $\phi_t(x)$ and $\psi_t(x)$
for $t=1,2$.

\subsubsection{Four coverage functions}
\begin{figure}
\includegraphics[width=0.5\textwidth]{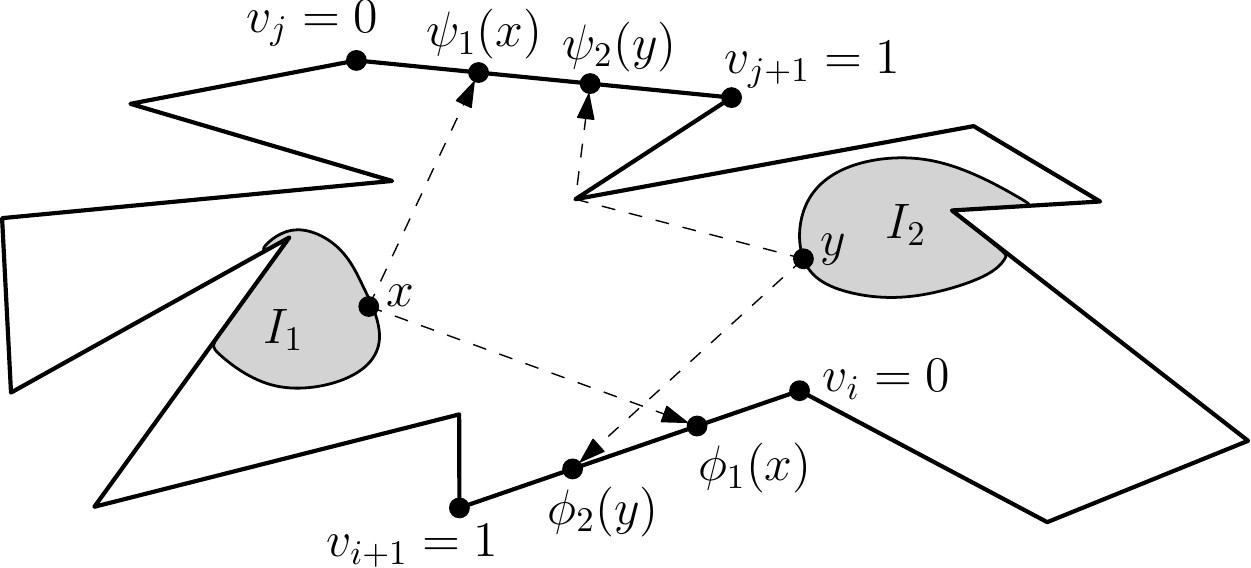}
\centering
\caption{\small
For a point
$x \in \partial I_t$, it holds that $d(\phi_t(x), x) = r$ and
$d(\psi_t(x), x) = r$ if $d(v_{i},x) \geq r$ and $d(v_{j+1},x) \geq r$
for $t=1,2$.
\label{fig:functions}}
\end{figure}

We represent each point $p\in e_i$
as a real number in 
 $[0,1]$: A point $x \in e_i$ (and $y \in e_j$)
 is represented as
$\|v_i-p\|/\|v_i-v_{i+1}\|$ in $[0,1]$, where
$\|x-y\|$ is the Euclidean distance between two
points $x$ and $y$.  Similarly, we represent each
point $q \in e_j$ as $\|v_j-q\|/\|v_j-v_{j+1}\|$.
We use a real number in $[0,1]$ and its
corresponding point 
interchangeably.  Recall that $\sarcs_1$ and
$\sarcs_2$ are the unions of the circular arcs of
$\partial I_1$ and $\partial I_2$, respectively.

Let us define the four functions $\phi_t(x)$ and
$\psi_t(x)$ for $t=1,2$ as follows.  Refer to
Figure~\ref{fig:functions}.
\begin{itemize}
\item The function $\phi_1 : \sarcs_1 \rightarrow
  [0,1]$ maps $x \in \sarcs_1$ to the infimum of
  the numbers which represent the points in
  $D_r(x) \cap e_i$.

\item The function $\phi_2: \sarcs_2 \rightarrow
  [0,1]$ maps $x \in \sarcs_2$ to the supremum of
  the numbers which represent the points in
  $D_r(x) \cap e_i$.

\item The function $\psi_1 : \sarcs_1 \rightarrow
  [0,1]$ maps $x \in \sarcs_1$ to the supremum of
  the numbers which represent the points in
  $D_r(x) \cap e_j$.

\item The function $\psi_2 : \sarcs_2 \rightarrow
  [0,1]$ maps $x \in \sarcs_2$ to the infimum of
  the numbers which represent the points in
  $D_r(x) \cap e_j$.
\end{itemize}

In the following, let $t$ be $1$ or $2$.  Our goal is to
split $\sarcs_t$ into subchains such that
$\phi_t$ and $\psi_t$ are monotone when their domain is restricted
to each subchain.  However, $\sarcs_t$ is not
necessarily a connected subset of $\partial I_t$.
Thus, to simplify the description of the split, we
define four continuous functions $\phi'_t,\psi'_t$
: $\partial I_t \rightarrow [0,1]$ by
interpolating $\phi_t$ and $\psi_t$ on $\partial
I_t$:
\[ \phi'_t(x) =
\begin{cases}
  \phi_t(x)& \quad \text{if } x \in \sarcs_t\\
  \frac{d_c(x_1,x)}{d_c(x_2,x_1)}\phi_t(x_1)+
  \frac{d_c(x_2,x)}{d_c(x_2,x_1)}\phi_t(x_2)
  & \quad \text{otherwise,}\\
\end{cases}
\]
where $x_1$ and $x_2$ are the first and the last
points of $S_t$ along $\partial I_t$ from $x$ in
clockwise order, 
respectively, and $d_c(x',y')$ denotes the length
of a chain $\subchain{x'}{y'}$.  The function
$\psi_t'$ is defined similarly.  
\begin{lemma}\label{extreme}
  The functions $\phi'_t$ and $\psi'_t$ for
  $t=1,2$ 
  are well-defined. 
\end{lemma}

\begin{proof}
  Here we prove the lemma only for $\phi'_1$. 
  For the other functions, the lemma can be proved analogously.

  If $\phi_1(x)$ is well-defined, so is $\phi_1'(x)$.  Thus, we
  show that $\phi_1(x)$ is well-defined. The set $\sarcs_1$ is
  the domain of $\phi_1$ and $e_i$ is the range of $\phi_1$.  Since
  $D_r(v_{i+1})$ contains all points in $\sarcs_1$, $D_r(x)$ 
  contains $v_{i+1}$ for all $x \in \sarcs_1$.  For each point $x \in
  \sarcs_1$, there are two cases: $D_r(x)$ intersects $e_i$ or
  contains $v_{i}$.  For the first case, $\phi_1(x)$ represents the
  point closest to $v_i$ among the points $p \in e_i$ with $d(x,p)=r$.
  For the second case, $\phi_1(x)$ is $1$, which represents $v_{i+1}$.
  Thus, $\phi_1(x)$ is uniquely defined for a point $x \in \sarcs_1$,
  which means that it is well-defined.
\end{proof}

We choose any two points $w_1 \in \partial I_1$ and $w_2 \in \partial
I_2$ which are endpoints of some circular arcs of $\partial I_1$ and
$\partial I_2$, respectively, such that $d(w_t, v_i) < r$ and $d(w_t,v_{j+1}) < r$.
Such points always exist by the assumption that $r(v_i,v_{j+1}) < r$.  
We use them as reference points for
$\partial I_1$ and $\partial I_2$.  
We write $p \prec q$ for any two points $p \in \partial I_t$ and $q
\in \partial I_t$, if $p$ comes before than $q$ when we traverse
$\partial I_t$ in clockwise order from the reference point $w_t$ for
$t= 1, 2$.  We consider $\partial I_1$ and $\partial I_2$ as chains of
circular and linear arcs starting from $w_1$ and $w_2$, respectively.

In the following, we consider the local extrema of the four functions.
For the function $\phi_t'$,  let  $N_{\max}$ be the 
set of points $x' \in I_t\setminus\{w_t\}$ such that $\phi_t'(x)$ has a local maximum at $x=x'$.
Similarly, let $N_{\min}$ be the 
set of points $x' \in I_t\setminus\{w_t\}$ such that $\phi_t'(x)$ has a local minimum at $x=x'$.
Then the following lemma holds.
\begin{lemma}
	 Both $N_{\max}$ and $N_{\min}$ are connected. 
\end{lemma}
\begin{proof}
We first consider the case that $\phi_1'(x) \neq 1$ for some $x\in N_{\max}$.
The boundary of $D_r(\phi_1'(x))$ intersects $\bd I_1$ at $x$.  
Moreover, there exists a connected region $N(x) \subset \bd I_1$ containing $x$ such that 
$D_r(\phi_1'(x)) \cap N(x) = \{x\}$.
Together with Lemma~\ref{pseudo}, 
this implies that $D_r(\phi_1'(x))$ does not contain any point other than
$x$.  Thus, $x$ is the only point contained in $N_{\max}$.

For the remaining case that $\phi_1'(x)=1$ for any point $x$ in $N_{\max}$,
Lemma~\ref{pseudo} implies that $D_r(\phi_1'(x)) \cap \partial I_1$ is connected.
By definition, $N_{\max} = D_r(\phi_1'(x))\cap \partial I_1$, which is connected.

Similarly, we can prove that $N_{\min}$ is connected.
\end{proof}

Let $x_{\max}$ and $x_{\min}$ be any two points in 
$N_{\max}$ and $N_{\min}$, respectively. 
To make the description easier, we assume that $x_{\max}
\prec x_{\min}$. 

The following corollary states Lemma~\ref{extreme} from a different
point of view.
\begin{corollary} \label{cor:extreme} The function $\phi_1$ is
  monotonically increasing in the domain $\{x \in \sarcs_1 : x \prec
  x_{\max} \}$ and in the domain $\{x \in \sarcs_1 : x_{\min} \prec
  x\}$ and monotonically decreasing in the domain $\{x \in \sarcs_1 :
  x_{\max} \prec x \prec x_{\min} \}$.
\end{corollary}

\subsubsection{Subdividing the chains with respect to local	extrema}
\label{subsubsect subdividing}
\label{sec:finding the local extrema}

We first compute one local maximum and one local minimum for each function.
Here we describe a way to find a local maximum of $\phi_1'$ lying on $S_1$.  
For a point $x$ lying on $S_1$, it holds that
$\phi_1'(x)=\phi_1(x)$ by definition.
Consider the sequence of the endpoints of the finer circular arcs
starting from the reference point $w_1$.  There are $O(n)$
endpoints.
First we choose the median $w$ of the endpoints. 
Let $w'$ be the
endpoint adjacent to $w$ on $S_1$ such that $w' \prec w$.  
Then we compute $\phi_1(w)$ and $\phi_1(w')$. This takes
$O(\log^2 n)$ time by the following lemma.
\begin{lemma}
	\label{computing function value}
	For a given point $x \in \sarcs_1$, $\phi_1(x)$ can be computed in
	$O(\log^2 n)$ time once the shortest path trees rooted at $v_i$ and
	$v_{i+1}$ are constructed.
\end{lemma}
\begin{proof}
	The function $d(z, x)$ for $z \in e_i$ is convex for a fixed point $x$ by
	Lemma~\ref{path-convex}.  
	Moreover, $d(v_i,x) \leq r$ and $d(v_{i+1},x) \geq r$ since $\phi_1$
	is well-defined.  As we saw before, $\mathcal{L}_i$ subdivides the
	edge $e_i$ into $O(n)$ subedges.  For any point $p$ on the same
	subedge, the combinatorial structure of the geodesic path $\pi(p,
	x)$ is the same.
	
	To compute $\phi_1(x)$, we apply binary search on $\mathcal{L}_i$.
	First, we choose the median $p_{\mathrm{med}}$ of $\mathcal{L}_i$.
	If $d(p_{\mathrm{med}},x) > r$, then $\phi_1(x)$ lies between $v_i$
	and $p_{\mathrm{med}}$, and we search the points in $\mathcal{L}_i$
	lying between them.  Otherwise, $\phi_1(x)$ lies between
	$v_{\mathrm{med}}$ and $v_{i+1}$.  In either way, we can ignore half
	of the current search space.  After $O(\log n)$ iterations, we can
	narrow the search space into the subedge containing $\phi_1(x)$.
	Once we find the subedge which contains $\phi_1(x)$, we can find the
	point in constant time.
	
	Since computing the geodesic distance between two points takes
	$O(\log n)$ time and the number of iterations is $O(\log n)$, the time
	complexity for computing $\phi_1(x)$ is $O(\log ^2 n)$ for a point
	in $\sarcs_1$.
\end{proof}

If $\phi_1(w)=1$, $w$ is a local maximum of $\phi_1'$ lying on $S_1$.
If $\phi_1(w_1) < \phi_1(w') < \phi_1(w)$, then a local maximum
comes after $w'$. Thus, we only consider the endpoints which come
after $w'$.  Otherwise, a local maximum
comes before $w$ by Corollary~\ref{cor:extreme}.  After $O(\log n)$
iterations, we can find the finer circular arc $s_{\max}$ in
$\sarcs_1$ which contains a local maximum point.

The remaining step is to find a local maximum point on the finer
circular arc $s_{\max}$.
Now, we search the edge $e_i$ to find the interval of $\mathcal{L}_i$ containing
a local maximum of $\phi'_1$.
Let $p$ be the median of $\mathcal{L}_i$.  If
$D_r(p)$ contains or intersects $s_{\max}$, we search further the
points of $\mathcal{L}_i$ which come after $p$.  Otherwise, we search
the points of $\mathcal{L}_i$ which come before $p$.  By the
construction of $\mathcal{L}_i$, the number of different combinatorial
structures of $\pi(x,p)$ for a point $x$ in the same circular arc and 
a point $p$ in the same subedge is at most
three (see Figure~\ref{fig:combinatorial_structure}).  Thus, in constant
time, we can check whether $D_r(p)$ contains or intersects $s_{\max}$.

After $O(\log n)$ iterations, we find the subedge that contains a
local maximum of $\phi_1'$.  Then we find a local maximum in the
finer circular arc in constant time.  Similarly, we compute a local
maximum and a local minimum for the other functions.

Therefore, we have the following lemma.
\begin{lemma}
	\label{lem:third-step}
	A local maximum and a local minimum for $\phi_t'$ (or $\psi_t'$) can be computed in
	$O(\log^3 n)$ time.
\end{lemma}

These local extrema subdivide $\partial I_1$ into at most five
subchains $c_{1,k}$ for $k \in \{1,2,\ldots, 5\}$ as follows. Let
$x_1, x_2, x_3$ and $x_4$ be the local maxima and the local minima of
$\phi_1'$ and $\psi_1'$ with $x_1 \prec x_2 \prec x_3 \prec x_4$.  The
subchain $c_{1,k}$ is the set of points $x \in \partial I_1$ with
$x_{k-1} \prec x \prec x_k$ for $k \in \{1,2,\ldots, 5\}$, where we
set $x_0 = x_5 = w_1$.  After subdividing $\partial I_1$, $\phi_1$ and
$\psi_1$ are monotone when the domain is restricted to $c_{1,k} \cap
\sarcs_1$ for $k \in \{1,2,\ldots, 5\}$.  Similarly, the local extrema
of $\phi_2'$ and $\psi_2'$ subdivide the chain $\partial I_2$ into five
subchains $c_{2,\ell}$ ($\ell \in \{1,2,\ldots, 5\}$ ).  The functions
$\phi_2$ and $\psi_2$ restricted to $c_{2,\ell} \cap \sarcs_2$ for
$\ell \in \{1,2,\ldots, 5\}$ are monotone.

\subsection{Computing a \texorpdfstring{$2$}{2}-center restricted to a
  pair of subchains}
\label{sec:computing a $2$-center for a pair of subchains}
We consider a pair $(c_{1,k}, c_{2,\ell})$ of subchains for $k \in
\{1,2,\ldots, 5\}$ and $\ell \in \{1,2,\ldots, 5\}$.  Let $s_{1,k} =
\sarcs_1 \cap c_{1,k}$ and $s_{2,\ell} = \sarcs_2 \cap c_{2,\ell}$.
We find a $2$-center with radius $r$ that is restricted to
$(e_i,e_j)$, if it exists, where one center is on $s_{1,k}$ and the
other is on $s_{2,\ell}$.  Assume that $\phi_1$ and $\psi_1$ are
decreasing when their domains are restricted to $s_{1,k}$.
That is, for any two points $x$ and $x'$ in $S_{1,k}$ with $x \prec x'$, 
it holds that $\phi_1(x') \leq \phi_1(x)$ and $\psi_1(x')\leq \psi_1(x)$.
Similarly,
assume that $\phi_2$ and $\psi_2$ are decreasing when their domains are
restricted to $s_{2,\ell}$.  The other cases where some functions are
increasing and the others are decreasing can be handled in a similar way.

We define two new functions $\mu_1 : s_{1,k} \rightarrow s_{2,\ell}$
and $\mu_2 : s_{1,k} \rightarrow s_{2,\ell}$.
For a point $x \in s_{1,k}$, $\mu_1(x)$ denotes the last clockwise point in
$s_{2,\ell}$ which is contained in $D_r(\phi_1(x))$.  Similarly, for a
point $x \in s_{1,k}$, $\mu_2(x)$ denotes the first clockwise point in
$s_{2,\ell}$ which is contained in $D_r(\psi_1(x))$.  If every point in
$s_{2,\ell}$ is contained in $D_r(\phi_1(x))$, then $\mu_1(x)$ is the last
clockwise point of $s_{2,\ell}$.  Notice that $\mu_1(x)$ and
$\mu_2(x)$ are increasing on $s_{1,k}$.  If there is a point $x
\in s_{1,k}$ such that $\mu_2(x) \prec \mu_1(x)$, the triplet $(x,
\mu_1(x),r)$ is restricted to $(e_i,
e_j)$. 
Moreover, for a 2-center $(c_1,c_2)$ restricted to $(e_i,e_j)$
with $c_1 \in s_{1,k}$ and $c_2 \in s_{2,\ell}$,
it holds that $\mu_2(c_1) \prec c_2 \prec \mu_1(c_1)$.
Thus, we are going to find a point $x \in s_{1,k}$ such that
$\mu_2(x) \prec \mu_1(x)$.

To check whether there exists such a point, we traverse $c_{1,k}$ twice. 
In the first traversal, we pick $O(n)$ points, which
are called \emph{event points}.  While picking such points, we compute $\mu_1(x)$ and
$\mu_2(x)$ for every event point $x$ in linear time.  Then we
traverse the two subchains again and find a $2$-center using the information
we just computed.

\paragraph{Definition of the event points on \texorpdfstring{$c_{1,k}$}{c}.}
We explain how we define the event points on $c_{1,k}$.  The set of
\emph{event points} of $c_{1,k}$ is the subset of $c_{1,k}$ consisting
of points belonging to one of the three types defined below.
\begin{itemize}
\item (T1) The endpoints of all finer arcs.
  Recall that the subchain $c_{1,k} \subseteq \partial I_1$ consists of
  circular arcs and line segments, and it is subdivided into
  finer arcs by the four shortest path maps in Section~\ref{subsection Subdividing the edges and the
    chains}.
\item (T2) The points $x \in s_{1,k}$ such that $d(x,p) = r$ for some
  $p \in \mathcal{L}_i$
\item (T3) The points $x \in s_{1,k}$ such that $d(x,p) = r$ for some
  $p \in \mathcal{L}_j$
\end{itemize}
Let $\mathcal{E}_1$, $\mathcal{E}_2$ and $\mathcal{E}_3$ be the sets
of event points of types T1, T2 and T3, respectively.  Let
$\mathcal{E} = \mathcal{E}_1 \cup \mathcal{E}_2 \cup \mathcal{E}_3$.
We say $\eta \in \mathcal{E}$ is \emph{caused} by $p$ if $d(\eta,p)=r$
for $p \in \mathcal{L}_i \cup \mathcal{L}_j$.

Recall that $\mathcal{L}_i$ is the set
of intersection points of the extensions of the edges in the two
shortest path trees rooted at $v_i$ and rooted at $v_{i+1}$ with
$e_i$, which has already been constructed in a previous step.  Let
$\mathcal{L}_i = \{v_i=p_1, \ldots, p_m=v_{i+1} \}$, where the points
are labeled in clockwise order from $v_i$.

\paragraph{Computation of the event points on \texorpdfstring{$c_{1,k}$}{c}.}
Since we already maintain the arcs of $c_{1,k}$ in clockwise order, we
already have $\mathcal{E}_1$.  
In the following, we show how to compute all T2 points.
In a similar way, we compute all T3 points.

Initially, $\mathcal{E}_2$ is set to be empty.
Assume that we have reached an event point $\eta \in
\mathcal{E}_1\cup\mathcal{E}_2$ and have already computed all T2
points on the subchain lying before $\eta$. Let $\eta'$ be the T1
point next to $\eta$.  We find all
T2 points on the subchain lying between $\eta$ and $\eta'$ by walking
the subchain from $\eta$ to $\eta'$ once.  If $\eta$ lies in $c_{1,k}\setminus s_{1,k}$, 
it is contained on $\bd P$. In this case, let $h(\eta)$ be the last T2 point in
$s_{1,k}$ in clockwise order with $h(\eta) \prec \eta$.  Otherwise,
let $h(\eta) = \eta$.  While computing all T2 points, we also compute
$\phi_1(h(\eta))$ and maintain $\pi(\eta,\phi_1(h(\eta)))$ for every
event point $\eta \in \mathcal{E}_1 \cup \mathcal{E}_2$.

We have two cases; the subchain connecting $\eta$ and $\eta'$ 
is contained in $\bd P$ or contained in a circular arc of $c_{1,k}$.
This is because $\eta'$ is a T1 point, an endpoint of
a finer arc.
To handle these cases, we need the following two lemmas.
\begin{lemma}
	\label{lem:constant-geodesic}
	Let $x_1$ and $x_2$ be any two points in the same finer arc of $c_{1,k}$.
	Once we have $\pi(x_1,p)$ for some point $p \in e_i$ and the finer arc
	of $c_{1,k}$ containing $x_1$ and $x_2$,
	we can compute $\pi(x_2,p)$ in constant time.
\end{lemma}
\begin{proof}
	Since we subdivide $I_1$ into finer arcs using the shortest path trees 
	rooted at $v_i$ and at $v_{i+1}$,
	$\pi(x_1,v_i)$ and $\pi(x_2,v_i)$ have the same combinatorial structure.
	Similarly, $\pi(x_1,v_{i+1})$ and $\pi(x_2,v_{i+1})$ have the same combinatorial
	structure.
	
	If $\pi(x_1,p)$ is not elementary, $\pi(x_1,p)$ and $\pi(x_2,p)$ have
	the same combinatorial structure. So, we can compute $\pi(x_2,p)$ in constant time.
	
	If $\pi(x_1,p)$ is elementary, $\pi(x_1,p)$ and $\pi(x_2,p)$
	may have distinct combinatorial
	structures. But in this case, $\pi(x_2,p)$ is also elementary.
	Moreover, it consists of a line segment, or two line segments whose
	common endpoint is the anchor of $\pi(v_i,x_2)$ or $\pi(v_{i+1},x_2)$
    closest to $x_2$.
	Note that the anchor of $\pi(v_i,x)$ closest to $x$ is the same for every point $x$ in the same
	cell in $\spm{i}$. We can compute this information while subdividing $\bd I_1$ into
	finer arcs. Thus, we may assume that we already have the anchors of $\pi(v_i,x)$ 
	and of $\pi(v_{i+1},x)$ closest to $x$.
	Thus, we can compute $\pi(x_2,p)$ in constant time.
	
	Therefore, in any case, we can compute $\pi(x_2,p)$ in constant time.
\end{proof}

\begin{lemma}
	\label{lem:constant-geodesic-edge}
	Let $p$ and $p'$ be any two points lying in the same subedge of $e_i$.
	Once we have $\pi(p,x)$ for some point $x \in c_{1,k}$
	and the finer arc of $c_{1,k}$ containing $x$,
	we can compute $\pi(p',x)$ in constant time.
\end{lemma}
\begin{proof}
	By the construction of $\mathcal{L}_i$, $\pi(p,v)$ and $\pi(p',v)$ 
	have the same combinatorial structure for any vertex $v$ of $P$.
	Thus, if $\pi(p,x)$ is not elementary, 
	$\pi(p,x)$ and $\pi(p',x)$ have the same combinatorial structure.
	
	If $\pi(p,x)$ is elementary, $\pi(p',x)$ is also elementary, but their
	combinatorial structures may be different.
	In this case, $\pi(p',x)$ consists of a line segment, or two line segments
	whose common endpoint is the anchor of  $\pi(v_i,x)$ or $\pi(v_{i+1},x)$ closest to $x$.
	We already have the cells of $\spm{i}$ and $\spm{i+1}$ containing $x$
	because we have the finer arc of $c_{1,k}$ containing $x$, which is the assumption
	of the lemma.
	Thus, we can compute $\pi(p',x)$ in constant time.
	
	Therefore, in any case, we can compute $\pi(p',x)$ in constant time.
\end{proof}

Now, we show how to handle the cases. Here, we assume that 
we already have $\pi(\eta,\phi_1(h(\eta)))$.

\paragraph{Case 1. The subchain is contained in \texorpdfstring{$\bd P$}{bd P}.}
If the subchain of $c_{1,k}$ connecting $\eta$ and $\eta'$ is
contained in $\partial P$, there is no T2 point lying between $\eta$ and $\eta'$.
Thus $\eta'$ is the event point next to $\eta$ and we simply 
compute $\phi_1(h(\eta'))$ and $\pi(\eta',\phi_1(h(\eta')))$. 

If $\eta' \notin s_{1,k}$, we have $h(\eta')=h(\eta)$.  We compute
$\pi(\eta',\phi_1(h(\eta)))$, which takes constant time by Lemma~\ref{lem:constant-geodesic}.  

If $\eta'\in s_{1,k}$, we have $h(\eta')=\eta'$.  
To compute $\pi(\eta',
\phi_1(\eta'))$, we first compute $\pi(\eta', p_{i'})$ and $d(\eta',p_{i'})$, where
$p_{i'}p_{i'+1}$ is the subedge of $e_i$ which contains
$\phi_1(h(\eta))$.  
They can be computed in constant time by Lemma~\ref{lem:constant-geodesic} and Lemma~\ref{lem:constant-geodesic-edge}.
Since $\phi_1$ is decreasing, $\phi_1(\eta')$ lies
on $\subchain{v_i}{\phi_1(h(\eta))}$. If $d(\eta', p_{i'}) > r$, then
$\phi_1(h(\eta))$ does not lie on
$\subchain{p_{i'}}{\phi_1(h(\eta))}$, and we skip $p_{i'}$.  We
check each subedge of $e_i$ from $p_{i'}$ in counterclockwise order
until we find the subedge containing $\phi_1(\eta')$. Then we compute
the geodesic path $\pi(\eta', \phi_1(\eta'))$ for $\phi_1(\eta')$ on
the subedge.  This takes time linear to the number of subedges
we traverse on $e_i$ by Lemma~\ref{lem:constant-geodesic} and
Lemma~\ref{lem:constant-geodesic-edge}. 

\paragraph{Case 2. The subchain is contained in a circular arc of \texorpdfstring{$c_{1,k}$}{c1k}.}
In this case, we first compute $\pi(\eta',p_{i'})$
and $d(\eta',p_{i'})$, where $p_{i'}p_{i'+1}$ is the subedge of $e_i$ which
contains $\phi_1(h(\eta))=\phi_1(\eta)$.  
This takes constant time by Lemma~\ref{lem:constant-geodesic} and Lemma~\ref{lem:constant-geodesic-edge}.

 If $d(\eta',p_{i'})$ is at least $r$, then $\phi_1(\eta')$ lies
between $\phi_1(\eta)$ and $p_{i'}$.  In this case, there is no T2 point lying between $\eta$ and
$\eta'$.  If $d(\eta',p_{i'})$ is less than $r$, then
there is an event point caused by $p_{i'}$ lying between $\eta$ and
$\eta'$.  It can be computed in constant time. Moreover, it is the
first T2 point from
$\eta$.
Then, we have to compute $\pi(\eta'', \phi(\eta_1(\eta'')))$ for the first T2 point $\eta''$ from $\eta$.
We can do this in constant time as we did for Case 1.

\paragraph{Definition and computation 
	of the event points on \texorpdfstring{$c_{2,\ell}$}{c}.}
The event points on $c_{2,\ell}$ are defined similarly.
Each event point is a point on $c_{2,\ell}$ belonging to one of the three types defined below.
\begin{itemize}
	\item (T1) The endpoints of all finer arcs of $c_{2,\ell}$.
	\item (T2) The points $x \in s_{2,\ell}$ such that $d(x,p) = r$ for some
	$p \in \mathcal{L}_i'$, where $\mathcal{L}_i'$ is the set of all points
	in $\mathcal{L}_i$ and all points $p \in e_i$ with $d(\eta,\phi_1(\eta))=r$ for
	some $\eta \in \mathcal{E}_1$. 
	\item (T3) The points $x \in s_{2,\ell}$ such that $d(x,p) = r$ for some
	$p \in \mathcal{L}_j'$, where $\mathcal{L}_j'$ is the set of all points
	in $\mathcal{L}_j$ and all points $p \in e_j$ with $d(\eta,\psi_1(\eta))=r$ for
	some $\eta \in \mathcal{E}_1$.
\end{itemize}
We already have $\mathcal{L}_i'$ and $\mathcal{L}_j'$, and the elements are sorted along
the edges $e_i$ and $e_j$, respectively.

The event points on $c_{2,\ell}$ can be computed in a way similar 
to the event points on $c_{1,k}$ in linear time.
Thus, we have the following lemma.
\begin{lemma}
	\label{lem:event-linear}
	The event points on $c_{1,k}$ and $c_{2,\ell}$ can be computed in $O(n)$ time.
\end{lemma}

\paragraph{Traversal for finding a restricted 2-center.}
Using the event points on $c_{1,k}$ and $c_{2,\ell}$, we can compute $\mu_1(x)$
and $\mu_2(x)$ for all $x \in \mathcal{E}$ in linear time.
Then, we can find a $2$-center restricted to
$(e_i,e_j)$ with radius $r$ by traversing $c_{1,k}$ as follows.
For every two consecutive event points $\eta, \eta'$ on $c_{1,k}$, 
we check whether there exists a point $x$ with $\eta \prec x \prec \eta'$ such that
$\mu_2(x) \prec \mu_1(x)$ using the following lemma.

\begin{lemma}
  Let $\eta$ and $\eta'$ in $\mathcal{E}$ be two consecutive event
  points along $s_{1,k}$.  We can determine whether there is a point
  $\eta \prec x \prec \eta'$ such that $\mu_2(x) \prec \mu_1(x)$ in
  time linear to the number of event points lying between
  $\mu_2(\eta)$ and $\mu_1(\eta')$ if $\mu_1(\eta)\prec \mu_2(\eta)
  \prec \mu_1(\eta') \prec \mu_2(\eta')$.
  Otherwise, we can determine whether there is such a point in constant time.
  \label{lemma:finding $2$-center between two event points}
\end{lemma}
\begin{proof}
  By the construction, $\phi_1(x)$ lies in the same subedge induced 
  by $\mathcal{L}_i$ for all $\eta \prec x \prec \eta'$, and so does
  $\phi_2(x)$ by $\mathcal{L}_j$.
  Moreover, $\mu_1(x)$ lies between $\mu_1(\eta)$ and $\mu_1(\eta')$, 
  and $\mu_2(x)$ lies between $\mu_2(\eta)$ and $\mu_2(\eta')$.
  
  Consider the case where $\mu_1(\eta)\prec \mu_2(\eta)
  \prec \mu_1(\eta') \prec \mu_2(\eta')$.   
  For a point $x$ with $\eta \prec x \prec \eta'$,
  it holds that $\mu_1(x) \prec \mu_2(x)$ if and only if
  there is a point $y$ with $\mu_2(\eta) \prec y \prec \mu_1(\eta')$
  and $\max\{d(\phi_1(x),y), d(\phi_2(x),y)\} \leq r$.  
  Note that for two consecutive event points $\nu$ and $\nu'$ on $s_{2,\ell}$,
  $d(\phi_1(x),y)$ and $d(\phi_2(x),y)$ are algebraic functions of
  constant degree for $\nu \prec y \prec \nu'$ and $\eta \prec x \prec
  \eta'$.  
  Moreover, we can find the algebraic functions while computing the event points.
  Thus, in constant time, we can determine whether there exists such a pair $(x,y)$ 
  such that $y$ lies between given two consecutive event points in $s_{2,\ell}$.
  
  We do this for every two consecutive event points lying between
  $\mu_2(\eta)$ and $\mu_1(\eta')$, which takes time linear to
  the number of event points lying between them.
  
  For the remaining case, we can answer ``yes'' or ``no'' in constant time.
  To see this, consider three possible subcases;
  $\mu_2(\eta) \prec \mu_1(\eta)$, $\mu_2(\eta') \prec \mu_2(\eta')$ or
  $\mu_1(\eta') \prec \mu_2(\eta)$.
  
  If $\mu_2(\eta) \prec \mu_1(\eta)$ or $\mu_2(\eta') \prec \mu_2(\eta')$,
  the answer is clearly ``yes.''
  If $\mu_1(\eta') \prec \mu_2(\eta)$,
  the answer is ``no'' because it holds that $\mu_1(\eta) \prec \mu_1(x)
  \prec \mu_1(\eta') \prec \mu_2(\eta) \prec \mu_2(x) \prec \mu_2(\eta')$
  for every point $x$ lying between $\eta$ and $\eta'$.
\end{proof}

If $r \geq \radrestricted$, there exists a $2$-set $(c_1,c_2)$ with
radius $r$ such that $c_1 \in S_1$ and $c_2 \in S_2$ by
Lemma~\ref{boundary is sufficient}.  We have $\mu_2(c_1) \prec c_2
\prec \mu_1(c_1)$.  Thus, the algorithm always find a 2-center with radius $r$.

We analyze the running time for traversing the chain $c_{1,k}$.
There are two types of a pair $(\eta, \eta')$ of consecutive event points;
$\mu_1(\eta)\prec \mu_2(\eta) \prec \mu_1(\eta') \prec \mu_2(\eta')$ or not.
The running time for handling
two consecutive event points belonging to the first case is linear to the number of event
points lying between $\mu_2(\eta)$ and $\mu_1(\eta')$.
For the second case, the running time is constant.

Here, for the pairs $(\eta, \eta')$ belonging to the first case,
their corresponding subchains $\{y : \mu_2(\eta) \prec y \prec \mu_1(\eta')\}$
are pairwise disjoint.
This implies that the total running time is linear to the number 
event points on $c_{1,k}$ and on $c_{2,\ell}$, which is $O(n)$.

\begin{lemma}
	Given two sets of all event points on $c_{1,k}$ and of all
	event points on $c_{2,\ell}$,
	a 2-center with radius $r$ restricted to $(e_i,e_j)$ can be computed in
	$O(n)$ time.
\end{lemma}

\subsection{The analysis of the decision algorithm}\label{analysis}

Now we analyze the running time of the algorithm.  
In the first step
described in Section~\ref{sec:constructing-intersections}, we
compute the intersection of the geodesic disks.
Once we have the farthest-point geodesic Voronoi diagram,
this step takes linear time.

In the second step described in Section~\ref{subsection Subdividing the edges and the chains}, 
we subdivide the edges and the boundaries of the intersections.
For subdividing the edges, we compute the four shortest path trees in linear time
~\cite{shortest-path-tree}, and compute the intersections
of the edges $e_i$ and $e_j$ with the extensions of the edges in the trees.
Then we sort them along the edges, which takes $O(n\log n)$ time.
For subdividing the boundaries of the intersections $I_1$ and $I_2$,
we traverse the shortest path map
along $\partial I_1$ (or $\partial I_2$) once, which takes $O(n)$ time
by Lemma~\ref{number of arcs}. 

In the third step described in Section~\ref{sec:coverage-functions}, 
we find one local minimum and one local maximum of 
each function. Then we subdivide
$\bd I_1$ and $\bd I_2$ into $O(1)$ subchains. 
This takes $O(\log^3 n)$ time by Lemma~\ref{lem:third-step}.

In the last step described in Section~\ref{sec:computing a $2$-center for a pair of subchains}, we consider $O(1)$ subchain pairs.  For a given
subchain pair $(c_{1,k}, c_{2,\ell})$, we compute all event points on
the subchains and traverse
the two subchains once. This takes linear time as we have shown.

Therefore, we have the following lemma.
\begin{lemma}
  For a candidate edge pair $(e_i, e_j)$ and a radius $r$, we can
  decide whether $r\geq r_{ij}^*$ in $O(n)$ time, once 
  we have $\mathcal{L}_i$ and $\mathcal{L}_j$ and the
  farthest-point geodesic Voronoi diagrams of the vertices of
  $\subchain{v_{j+1}}{v_i}$ and of the vertices of
  $\subchain{v_{i+1}}{v_j}$ are computed. 
  In the same time, if
  $r\geq r_{ij}^*$, we can compute a $2$-center with radius $r$
  restricted to $(e_i, e_j)$.
\end{lemma}
Here, we do not consider the time for computing $\mathcal{L}_i$ and $\mathcal{L}_j$
and the farthest-point geodesic Voronoi diagrams because they do not depend
on input radius $r$.
In the overall algorithm, this decision algorithm will be executed repeatedly
with different input radius $r$.
In this case, we do not need to recompute 
$\mathcal{L}_i$ and $\mathcal{L}_j$
and the farthest-point geodesic Voronoi diagrams.

\section{An optimization algorithm for a candidate edge pair}

%
%
%

The geodesic $1$-center of a simple polygon is determined by at most
three convex vertices of $P$ that are farthest from the center. For a
given geodesic $2$-center $(c_1^*, c_2^*)$ with radius
$r^*=r(c^*_1,c^*_2)$, a similar argument applies.

Lemma~\ref{partition} and its proof imply that there are three
possible configurations for a $2$-center as follows. Let $\alpha^*$
and $\beta^*$ be the two endpoints of $b(c_1^*,
c_2^*)$. 
\begin{enumerate}
\item $d(c_t^*,\alpha^*)<r^*$ and $d(c_t^*,\beta^*)<r^*$ for $t=1,2$.
\item Either $d(c_1^*,\alpha^*)=d(c_2^*,\alpha^*)=r^*$ or
  $d(c_1^*,\beta^*)=d(c_2^*,\beta^*)=r^*$.
\item
  $d(c_1^*,\alpha^*)=d(c_2^*,\alpha^*)=d(c_1^*,\beta^*)=d(c_2^*,\beta^*)=r^*$.
\end{enumerate}

For Configuration 1, a $2$-center $(c_1^*, c_2^*)$ restricted to $(e_i, e_j)$ can be
computed in $O(n)$ time because $c_1^*$ is the 1-center of $\subpolygon{v_{i+1}}{v_{j}}$
and $c_2^*$ is the 1-center of $\subpolygon{v_{j+1}}{v_i}$. 
Thus we only focus on Configurations 2 and 3.

In this section, we present an algorithm for computing a 2-center
restricted to a given candidate edge pair $(e_i,e_j)$.
We apply the parametric searching technique \cite{parametric} to 
extend the decision algorithm
in Section~\ref{sec:decision algorithm} into an optimization algorithm.
We use the decision algorithm for two different purposes.
We simulate the decision algorithm 
with the optimal solution $r_{ij}^*$ (without explicitly computing $r_{ij}^*$.)
While simulating the decision algorithm with $r_{ij}^*$, we use the decision
algorithm as a subprocedure with an explicit input $r$.

In the following, we show how to simulate the decision algorithm
with the optimal solution $r_{ij}^*$, and finally compute the optimal solution.

\subsection{Constructing the intersections of geodesic disks}
\label{sec:constructing-intersections}
\begin{figure}
  \begin{center}
    \includegraphics[width=0.6\textwidth]{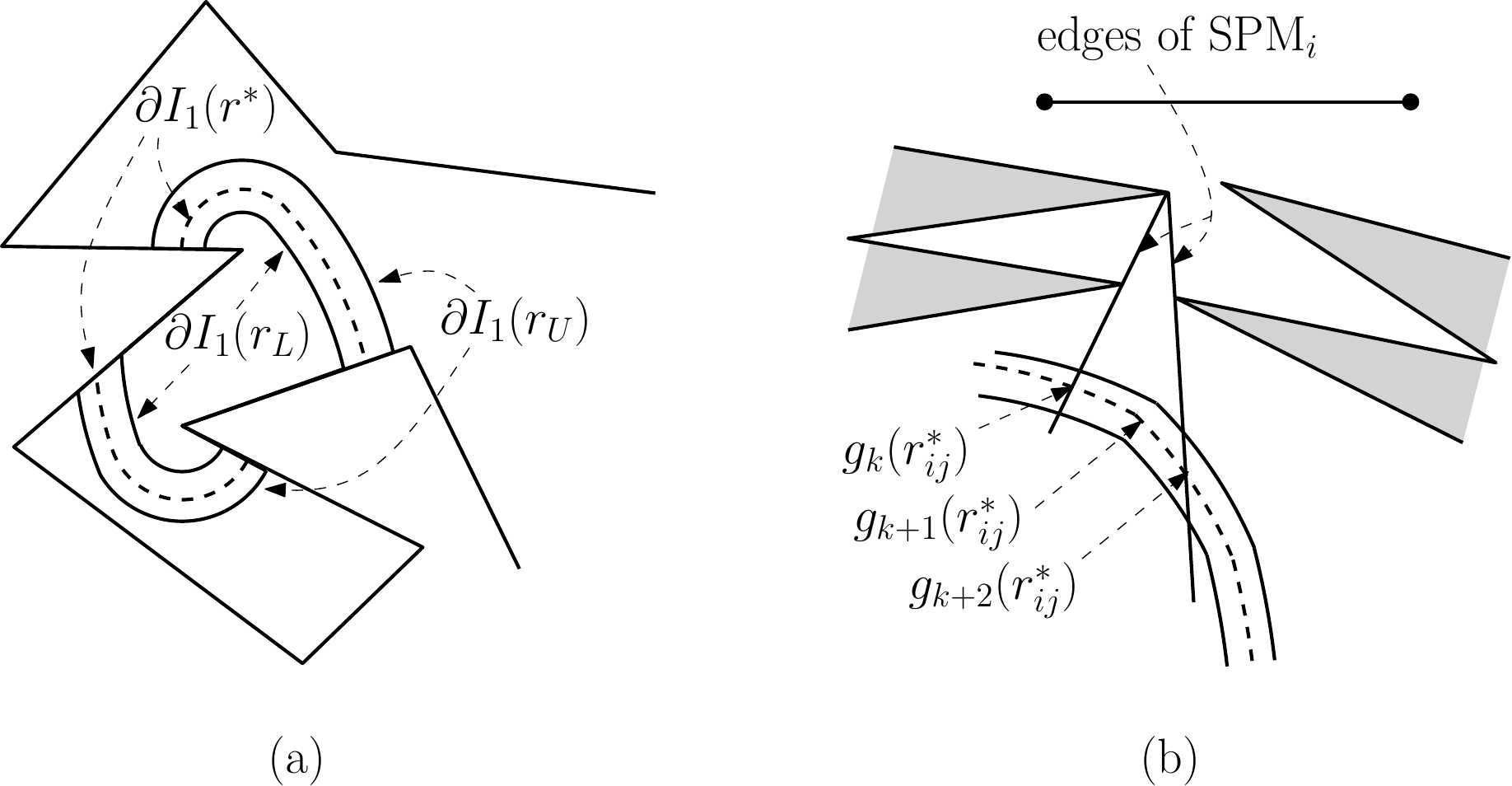}
    \caption{\small 
      (a) For all $r_L \leq r \leq r_U$, the combinatorial structures of $\partial I_1(r)$ are the same.
	 (b) The points 
	 $g_k(\radrestricted)$ and $g_{k+2}(\radrestricted)$ are intersections of $\bd I_1(\radrestricted)$
	 and edges of $\spm{i}$. The point $g_{k+1}(\radrestricted)$ is 
	 an endpoint of a finer arc of $\bd I_1(\radrestricted)$.
         \label{fig:parametric searching}}
     \end{center}
   \end{figure}

   We compute the farthest-point geodesic Voronoi diagrams, denoted by
   $\mathsf{FV}_1$ and $\mathsf{FV}_2$, of the vertices in
   $\subchain{v_{i+1}}{v_j}$ and $\subchain{v_{j+1}}{v_i}$ in
   $O(n\log\log n)$ time, respectively, as we did in the decision
   algorithm (refer to Section~\ref{sec:decision_intersection}).
   Let $I_1(r) = \cap_{k=i+1}^{j} {D_{r}(v_k)}$ and $I_2(r) =
   \cap_{k={j+1}}^{i} {D_{r}(v_k)}$.  Instead of computing
   $\partial I_1(\radrestricted)$ and $\partial I_2(\radrestricted)$
   explicitly, we compute the combinatorial structures of $\partial
   I_1(\radrestricted)$ and $\partial I_2(\radrestricted)$.  Here, the
   combinatorial structures of $\partial I_1(\radrestricted)$ and
   $\partial I_2(\radrestricted)$ are the cyclic sequences of edges of
   the farthest-point geodesic Voronoi diagram intersecting $\partial
   I_1(\radrestricted)$ and $\partial I_2(\radrestricted)$ in
   clockwise order, respectively.

   For each vertex $v_f$ in $\mathsf{FV}_1$ and $\mathsf{FV}_2$, we
   compute $d(v_f, t)$ for all sites $t$ of the cells incident to
   $v_f$.  Let $\mathcal{R}$ be the set of these distances.  Then we
   have $|\mathcal{R}| = O(n)$.  We sort those distances in increasing
   order and apply binary search on $\mathcal{R}$ to find the largest value $r_U\in \mathcal{R}$
   and the smallest values $r_L \in \mathcal{R}$ satisfying $\radrestricted\in [r_L,r_U]$
   using the decision algorithm.  
   We have already constructed $\mathsf{FV}_1$ and $\mathsf{FV}_2$.
   And we compute $\mathcal{L}_i$ and $\mathcal{L}_j$, which are
   independent of $r$. Then the
   decision algorithm takes linear time. Thus, this procedure takes
   $O(n \log n)$ time.
   
   Then, for any radius $r\in [r_L, r_U]$, the combinatorial structure
   of $\partial I_1(r)$ is the same.  See Figure~\ref{fig:parametric
     searching}(a).  Thus, by computing $\partial I_1(r_L)$, we can
   obtain the combinatorial structure of $\partial
   I_1(\radrestricted)$.  Note that each endpoint of the arcs of $\bd
   I_1(r)$ can be represented as algebraic functions of $r$ with
   constant degree for $r \in [r_L, r_U]$.

\subsection{Subdividing the intersections of geodesic disks}
In the following, we let $t$ be $1$ or $2$.  We subdivide $\partial
I_t(\radrestricted)$ by overlaying $\spm{i}$, $\spm{i+1}$, $\spm{j}$,
and $\spm{j+1}$ with $\partial I_t(\radrestricted)$.  Instead of
computing it explicitly, we compute the combinatorial structure of the
subdivision of $\partial I_t(\radrestricted)$.

First we compute the shortest path map $\spm{i}$.  The annulus
$I_t(r_U) \setminus \mathrm{int}(I_t(r_L))$ does not contain any
vertex of the shortest path map.  Thus the edges of $\spm{i}$
intersecting the curve $\partial I_t(r)$ can be
ordered 
along the curve $\partial I_t(r_U)$ in clockwise fashion for any $r
\in [r_L,r_U]$.  Moreover, the order of these edges 
is the same for any $r \in [r_L,r_U]$.  Thus, this order can be computed in linear time by
traversing $\partial I_t(r_U)$ and $\spm{i}$ as we did in
Section~\ref{sec:finding the local extrema}.  We do this also for
$\spm{i+1}$, $\spm{j}$, and $\spm{j+1}$.

Consider an edge $e$ of $\spm{i}$ intersecting $\bd I_t(r)$ for $r \in
[r_L,r_U]$.  Then the intersection can be represented as an algebraic
function of $r$.  We compute such an algebraic function for each edge
of $\spm{i}$, $\spm{i+1}$, $\spm{j}$, and $\spm{j+1}$ and merge them
with the set of the endpoints of the arcs in $\partial I_t(r)$ for $r
\in [r_L,r_U]$ computed in
Section~\ref{sec:constructing-intersections}.  Let $\mathcal{A}$
denote the merged set. There are $O(n)$ elements in $\mathcal{A}$ each
of which is an algebraic function of the variable $r$. 
They can be
sorted in clockwise order along $\partial I_t(\radrestricted)$
in $O(n \log^2 n)$ time by Lemma~\ref{lemma:sorting}.  See
Figure~\ref{fig:parametric searching}.  Let $\mathcal{A} = \{g_1(r),
\ldots, g_{m}(r)\}$.  The elements are the endpoints of \emph{finer}
arcs of $\partial I_t(r)$.  Let $\{g'_1(r), \ldots, g'_{m'}(r)\}$ be
the sorted list of the endpoints of finer arcs of $\partial I_2(r)$
for $r \in [r_L, r_U]$.

\begin{lemma}
  \label{lemma:sorting}
  The sorted list $\mathcal{A}$ can be computed
  in $O(n \log^2 n)$ time.
\end{lemma}
\begin{proof}
  As described above, we can compute all elements of $\mathcal{A}$ 
  in linear time.  In the following, we show that they can be
  sorted in $O(n\log^2 n)$ time.  Sorting $O(n)$ elements can
    be done in $O(T_\textnormal{c}\log n)$ time using $O(n)$
    processors \cite{parallel-sorting}, where $T_\textnormal{c}$ is
    the time for comparing two elements.  To compare two
  elements, that is, to determine the order for the two elements along 
  $\bd I_t(\radrestricted)$ with respect to the reference point for $\bd I_t(\radrestricted)$,
  we do the followings.
  Let $h_1(\radrestricted)$ and $h_2(\radrestricted)$ be two elements of $\mathcal{A}$.
  Here, $h_1(r)$ and $h_2(r)$ are algebraic functions of variable $r \in \mathbb{R}$,
  and we want to determine the order for them when $r=\radrestricted$.
  We first find the roots of $h_1(r)=h_2(r)$.
  Let $c$ be the number of roots, which is a
  constant.  We apply the decision algorithm $c$ times with the roots.
  Then we can compare $h_1(\radrestricted)$ and $h_2(\radrestricted)$
  in $O(n)$ time, which is the time for applying the decision
  algorithm $c$ times.  In other words, $T_\textnormal{c}=O(n)$ for
  our case.  Note that we have already constructed the farthest-point
  geodesic Voronoi diagrams.

  We apply parametric search~\cite{parametric} with this parallel
  sorting algorithm.  In each iteration of the sorting, we need to do
  $O(n)$ comparisons, which are done in different processors in the
  parallel sorting algorithm. That means each of them is independent
  of the others.  Thus, in each iteration, we find $O(n)$ roots for
  all functions.  Then we sort them and apply binary search on $O(n)$
  roots using the decision algorithm.  Each iteration takes $O(n \log
  n)$ time, thus the total time complexity is $O(n \log^2 n)$ time.
\end{proof}

This section can be summarized as follows.

\begin{lemma}
  The set of endpoints of the finer arcs of $\partial I_1(r)$ and
  $\partial I_2(r)$ can be computed in $O(n \log ^2 n)$ time for $r
  \in [r_L,r_U]$.
\end{lemma}

\subsection{Computing the coverage function values}

Recall that the function $\phi_1(p)$ maps a point $p \in P$ to the
infimum of the numbers which represent the points in $D_r(p) \cap
e_i$.
We find the subedge that contains $\phi_1(g_k(\radrestricted))$ for
each index $k \in [1,m]$ in $O(n \log ^2 n)$ time.  In the decision
algorithm, this can be done in $O(n)$ time.  However, each comparison
in the decision algorithm depends on the results of the comparisons in
the previous steps, thus this algorithm cannot be parallelized.
Therefore, we devise an alternative algorithm which can be
parallelized by allowing more comparison steps.

\begin{lemma}
  For an index $k \in [1,m]$, the subedge of $e_i$ containing
  $\phi_1(g_k(\radrestricted))$ can be computed in $O(n \log n)$ time.
\end{lemma}

\begin{proof}
  For a fixed $r$, we can compute $\phi_1(g_k(r))$ in $O(\log^2 n)$
  time by Lemma~\ref{computing function value}.  However, since we do
  not know the exact value $\radrestricted$, we cannot apply the
  algorithm of Lemma~\ref{computing function value}.  Instead, we
  again apply parametric search.

  The first step of the algorithm for finding $\phi_1(x)$ in
  Lemma~\ref{computing function value} is to compute
  $d(p_{\mathrm{med}}, x)$, where $p_{\mathrm{med}}$ is the median of
  $\mathcal{L}_i$.  Since $g_k(r)$ is contained in the same cell of
  $\spm{i}$ and in the same cell of $\spm{i+1}$ for all $r \in [r_L,
  r_U]$, the combinatorial structures of $\pi(p_{\mathrm{med}},
  g_k(r))$ are the same for all $r \in [r_L, r_U]$.  To determine
  whether $p_{\mathrm{med}}$ comes before or after
  $\phi_1(g_k(\radrestricted))$ from $v_i$ in $\mathcal{L}_i$, we
  check the sign of the value $\radrestricted-d(p_{\mathrm{med}},
  g_k(\radrestricted))$.  If it is positive, then $p_{\mathrm{med}}$
  comes before $\phi_1(g_k(\radrestricted))$.  If it is negative, then
  $p_{\mathrm{med}}$ comes after $\phi_1(g_k(\radrestricted))$.
  Otherwise, $p_{\mathrm{med}}$ equals $\phi_1(g_k(\radrestricted))$.
  We compute the roots of $r-d(p_{\mathrm{med}}, g_k(r))=0$ and apply
  the decision algorithm in Section~\ref{sec:decision algorithm} to
  the roots.  Since the farthest-point geodesic Voronoi diagrams have
  already been constructed, the running time for each call of the
  decision algorithm is $O(n)$.  Then we can determine the sign of the
  value in $O(n)$ time even though we still do not know the exact
  value $\radrestricted$.
  
  After repeating this step $O(\log n)$ times, we can finally find the
  subedge of $e_i$ containing $\phi_1(g_k(\radrestricted))$ in $O(n
  \log n)$ time.
\end{proof}

Since the total number of indices is $m = O(n)$, we can compute the
subedge of $e_i$ containing $\phi_1(g_k(\radrestricted))$ for all
indices $k \in [1,m]$ in $O(n^2 \log n)$ time.  To compute them
efficiently, we parallelize this procedure using $O(n)$ processors.  A
processor is assigned to each index.  In each iteration, we compute
the roots of $r-d(p_k,g_k(r))=0$ for all indices $k \in [1,m]$ and sort them,
where $p_k$ is the median of the current search space of an index $k$.
There are $O(n)$ roots. We apply binary search on the roots using our
decision algorithm and find the interval that containing
$\radrestricted$.  Then in $O(n \log n)$ time, we can finish the
comparisons in each iteration as we did in the proof of
Lemma~\ref{lemma:sorting}.  We need $O(\log n)$ iterations to find the
subedges, so we can compute the subedges for all indices in $O(n
\log^2 n)$ time.  We do this also for $\psi_1(x)$.  Similarly, we
compute the subedges containing the function values $\phi_2(x),
\psi_2(x)$ for all endpoints $x$ in $\partial I_2(\radrestricted)$
in $O(n \log^2 n)$ time.

Then, we compute the algebraic functions $\phi_1(g_k(r))$ and
$\phi_2(g'_{k'}(r))$ for all indices $k \in [1,m]$ and all indices $k' \in [1,m']$.
Then we sort the points in $\mathcal{L}_i$ and the points
$\phi_1(g_k(\radrestricted)), \phi_2(g'_{k'}(\radrestricted))$ for all indices
$k \in [1,m]$ and all indices $k' \in [1,m']$ in $O(n \log^2 n)$ time using
a way similar to the algorithm in Lemma~\ref{lemma:sorting}.

\begin{lemma}
  \label{sorting endpoints}
  The points $\phi_1(g_k(\radrestricted))$ and
  $\phi_2(g'_{k'}(\radrestricted))$ for all indices $k \in [1,m]$ and
  $k' \in [1,m']$, and the points in $\mathcal{L}_i$ can be sorted in
  $O(n \log^2 n)$ time.
\end{lemma}

\subsection{Constructing quadruples consisting of two cells and two
  subedges}
Consider a quadruple $(x_1,x_2,y_1,y_2)$, where $x_t$ is a finer arc of
$\partial I_t(\radrestricted)$ for $t=1,2$, and $y_1$ and $y_2$ are
subedges in $e_i$ and $e_j$, respectively.  We say the quadruple
$(x_1,x_2,y_1,y_2)$ is \emph{optimal} if there is a $2$-center
$(c_1,c_2)$ such that $c_1 \in x_1, c_2 \in x_2$ and $\pu \in y_1, \pw
\in y_2$ for some point-partition $(\pu,\pw)$ with respect to $(c_1,
c_2,\radrestricted)$.  
Given an optimal quadruple, we can
compute $c_1$ and $c_2$ in constant time by the following lemma.

\begin{lemma}
  Given an optimal $4$-tuple $(x_1,x_2,y_1,y_2)$, a $2$-center
  $(c_1,c_2)$ restricted to the candidate edge pair $(e_i,e_j)$ can be
  computed in constant time.
\end{lemma}

\begin{proof}
  Let $(c_1^*, c_2^*)$ be the $2$-center corresponding to the given
  $4$-tuple.  Consider the subdivision $\mathcal{M}_t$ which is the
  overlay of the graph obtained by extending the edges of the shortest
  path trees rooted at ${v_i}$, ${v_{i+1}}$, ${v_j}$, ${v_{j+1}}$ in
  both directions with $\mathsf{FV}_t$ for $t=1,2$.  (We do not
  construct $\mathcal{M}_1$ and $\mathcal{M}_2$ explicitly.  We
  introduce these subdivisions to make it easier to understand the
  analysis of our algorithm.)  By the construction of finer arcs and
  subedges, 
  $x_t$ is contained in a cell of $\mathcal{M}_t$ for $t=1,2$.
  Moreover, each endpoint of a finer arc lies either on an edge of a
  shortest path map or on an edge of a farthest-point geodesic Voronoi
  diagram.  Let $f_t$ be the site of the cell of $\mathsf{FV}_t$
  containing $x_t$.   Let $w_1^*,w_2^*,w_3^*$ and $w_4^*$ be
  the points that minimize the following function
  $g(w_1,w_2,w_3,w_4)$ for $w_1 \in x_1, w_2 \in x_2, w_3 \in y_1$,
  and $w_4 \in y_2$.
  \begin{linenomath}
    \begin{equation*}
 	g(w_1,w_2,w_3,w_4) =\max\{d(w_1,f_1),d(w_1,w_3),d(w_1,w_4),d(w_2,f_2),d(w_2,w_3),d(w_2,w_4)\}
     \end{equation*}  
  \end{linenomath}
  Then $(w_1^*,w_2^*)$ is the $2$-center restricted to $(e_i,e_j)$.

  Since each element of the quadruple is fully contained in some cell of
  $\mathcal{M}_1$ or $\mathcal{M}_2$, $g(w_1,w_2,w_3,w_4)$ is the maximum of
  the six algebraic functions of constant degree.  Thus we can
  find the minimum of $g(w_1,w_2,w_3,w_4)$ in constant time.
\end{proof}

However, there are more than quadratic number of such
quadruples. 
Instead of considering all of them, we construct a set of quadruples
with size $O(n)$ containing at least one optimal quadruple as follows.

We first compute the event points on $\partial I_t(\radrestricted)$.
An event point on $\partial I_t(\radrestricted)$ is a point $x$ with
$d(x, p) = \radrestricted$ for some point $p \in \mathcal{L}_i$.  We
do not need to compute the exact positions for event points.  Instead,
we compute their relative positions, i.e., we implicitly maintain those
points sorted in clockwise order along $\partial I_t(\radrestricted)$.
We can do this in $O(n \log^2 n)$ time as we did in the proof of
Lemma~\ref{sorting endpoints}.

Now we have all the event points that subdivide the
chain $\partial I_1(\radrestricted)$ and $\partial
I_2(\radrestricted)$ into $O(n)$ subarcs.  Moreover, we have
subdivided the edges $e_i$ and $e_j$ into $O(n)$ subedges (refer to
Section~\ref{sec:computing a $2$-center for a pair of subchains}).
Then, we construct the set of quadruples $(x_1,x_2,y_1,y_2)$ such that
there are event points $p(r) \in x_1$ and $q(r) \in x_2$, which are
indeed algebraic functions, satisfying $\phi_1(p(\radrestricted)) \in
y_1$, $\psi_1(p(\radrestricted)) \in y_2$, $\phi_2(q(\radrestricted))
\in y_1$ and $\psi_2(q(\radrestricted)) \in y_2$.  Note that each
event point $p(r)$ lies in the same cell of $\mathcal{M}_t$ for all $r
\in [r_L, r_U]$ by the construction.  Using the information we
computed before, we can construct the set of those quadruples in time
linear to the size of the set, which is $O(n)$,

Moreover, since we consider all quadruples one of which is an optimal
$4$-tuple, we can find a $2$-center restricted to $(e_i, e_j)$ using
the procedure in this section in $O(n \log^2 n)$ time.  The following
lemma and theorem summarize the our result.

\begin{lemma}
  A $2$-center restricted to a candidate edge pair can be computed in
  $O(n \log ^2 n)$ time.
\end{lemma}

\begin{theorem}
 For a simple polygon $P$ with $n$ vertices, a $2$-center of $P$
 can be computed in $O(n^2 \log ^2 n)$ time.
\end{theorem}

\end{document}